\documentclass{article}

\usepackage{epsfig,xspace}
\usepackage{latexsym}
\usepackage[matrix,arrow,curve]{xy}\xyoption{all}

\usepackage[usenames,dvipsnames,svgnames,table]{xcolor}
\usepackage{amssymb,amsmath,latexsym,amsfonts,amsthm,alltt}
\usepackage[usenames,dvipsnames,svgnames,table]{xcolor}
\usepackage[dvips,usenames]{stmaryrd}
\usepackage{enumitem}
\usepackage{url}
\usepackage{graphicx}
\usepackage{float}
\usepackage{amstext}
\usepackage{cite}
\usepackage{hyperref}
\usepackage{MnSymbol}

\newtheorem{defn}{Definition}[section]
\newtheorem{rem}[defn]{Remark}
\newtheorem{thm}[defn]{Theorem}
\newtheorem{lemma}[defn]{Lemma}
\newtheorem{prop}[defn]{Proposition}

\newcommand\crule[3][black]{\textcolor{#1}{\rule{#2}{#3}}}

\newcommand{\ra}{\rightarrow}
\newcommand{\lra}{\longrightarrow}

\newcommand{\Ra}{\Rightarrow}

\newcommand{\Da}{\Downarrow}

\newcommand{\onto}[1]{\stackrel{#1}{\leadsto}}

\newcommand{\midsp}{\;|\;}

\newcommand{\telos}{\hfill$\Box$}
\newcommand{\type}[1]{{\tt #1}}

\newcommand{\iso}{\backsimeq}

\newcommand{\val}[1]{\mbox{$[\![#1]\!]$}}
\newcommand{\forces}{\Vdash}
\newcommand{\dforces}{\forces^{\!\!\partial}}
\newcommand{\yvval}[1]{\mbox{$(\!|#1|\!) $}}

\newcommand{\infrule}[2]{\frac{\mbox{\rm $#1$}}{\mbox{\rm $#2$}}}
\newcommand{\proves}{\vdash}

\newcommand{\vmodels}{\mbox{$\medvert\!\!\!\!\approx\;$}}

\newcommand{\upv}{\upVdash}
\newcommand{\rperp}{\mbox{${}^{\upv}$}}
\newcommand{\gphi}{{\mathcal  G}(Z_\partial)}
\newcommand{\gpsi}{{\mathcal  G}(Z_1)}

\newcommand{\bbox}{\blacksquare}

\newcommand{\lperp}{{}\rperp}

\newcommand{\ldd}{\mbox{$\largediamond\hspace*{-10pt}\Diamond\;$}}

\newcommand{\lbvert}{\mbox{\Large \mbox{$\boxvert$}}}
\newcommand{\lbminus}{\mbox{\Large \mbox{$\boxminus$}}}

\newcommand{\ldvert}{\raisebox{0.5pt}{\Large \mbox{$\diamondvert$}}}

\newcommand{\ldminus}{\raisebox{0.5pt}{\Large \mbox{$\diamondminus$}}}
\newcommand{\dd}{\diamonddiamond}

\newcommand{\lbb}{\mbox{\raisebox{1pt}{$\largesquare\hspace*{-7.6pt}\Box\;$}}}
\newcommand{\bb}{\boxbox}

\newcommand{\lbbox}{\raisebox{-1.2pt}[0pt][0pt]{\crule[black]{0.27cm}{0.27cm}}\hspace*{1pt}}

\newcommand{\filt}{\mbox{\rm Filt}}

\newcommand{\idl}{\mbox{\rm Idl}}

\newcommand{\rfspoon}{\rightfilledspoon}

\newcommand{\tright}{\triangleright}

\newcommand{\Mtright}{\mathrel{\mbox{$|\!\!\largetriangleright$}}}

\newcommand{\dfootl}{\downfootline}
\newcommand{\ufootl}{\upfootline}

\title{Distribution-Free Normal Modal Logics}

\author{Chrysafis Hartonas\\
University of Thessaly, Greece; hartonas@uth.gr}

\begin{document}
\maketitle
\begin{abstract}
  This article initiates the semantic study of distribution-free normal modal logic systems, laying the semantic foundations and anticipating further research in the area. The article explores roughly the same area, though taking a different approach, with a recent article by Bezhanishvili, de Groot, Dmitrieva and Morachini, who studied a distribution-free version of Dunn's Positive Modal Logic (PML). Unlike PML, we consider logics that may drop distribution and which are equipped with both an implication connective and modal operators.   
We adopt a uniform relational semantics approach, relying on recent results on representation and duality for normal lattice expansions.  We prove canonicity and completeness in the relational semantics of the minimal distribution-free normal modal logic, assuming just the K-axiom, as well as of its axiomatic extensions obtained by adding any of the  D, T, B, S4 or S5 axioms.  Adding distribution can be easily accommodated and, as a side result, we also obtain a new semantic treatment of Intuitionistic Modal Logic.  
\end{abstract}

\section{Introduction}
\label{intro}
The first, other than \texttt{IML} (Intuitionistic Modal Logic), sub-classical modal logic system studied was Dunn's Positive Modal Logic (\texttt{PML}) \cite{pml}, weakening the lattice-base axioms of the logic to those of a mere distributive lattice. Both negation and implication are absent from the language of \texttt{PML}. Dunn's article has attracted some interest. Gehrke, Nagahashi and Venema \cite{Gehrke-Venema} studied the Sahlqvist theory of a system \texttt{DML} (Distributive Modal Logic), which is the fusion of its two monomodal fragments and is weaker than Dunn's \texttt{PML}. Celani and Jansana \cite{celani-jansana} proposed a new semantics for \texttt{PML}, to model the consequence relation of the logic. Finally, and more relevant to our purposes, Bezhanishvili, Dmitrieva, de Groot and Morachini \cite{choice-free-dmitrieva-bezanishvili}  focused on a distribution-free version of \texttt{PML}. A single relation $R$ in frames generates both $\Box,\Diamond$, the interpretation for modal formulae is standard, except that in pursuing a way to define frames that validate interaction axioms they find it necessary to abandon distribution of diamonds over joins. The authors do not work with canonical extensions, they model their logic(s) of interest in (topological) semilattices and they develop a related correspondence theory. A duality for distribution-free modal lattices is presented in \cite{choice-free-dmitrieva-bezanishvili}, building on the Hofmann, Mislove and Stralka \cite{hms-duality} duality for meet semilattices, and the authors use it to build a new semantic framework for the logic of their focus.

The logics we consider may, or may not be distributive and, axiomatically, they are the sub-classical analogues of the classical basic systems  of normal modal logic. Their equivalent algebraic semantics is given by the variety of implicative modal bounded lattices $\mathbf{L}=(L,\leq,\wedge,\vee,0,1,\ra,\Box,\Diamond)$. The implication fragment of the logics  may be as weak as the same fragment of the non-associative full Lambek calculus and as strong as the Intuitionistic (or Boolean) propositional calculus. We adopt a uniform relational semantics approach, based on a generalization of the J\'{o}nsson-Tarski framework \cite{jt1} for the relational representation of normal lattice operators given in \cite{duality2}.

The relational semantics of distribution-free systems is typically given in two-sorted frames \cite{Suzuki-polarity-frames,mai-gen,mai-grishin,choiceFreeStLog,vb}, with a distinguished relation generating a Galois connection on the powersets of the sorts and modeling sentences as sets that are stable under the closure operator generated by the Galois connection. It has been observed that a two-sorted framework makes it hard to transfer and use results and techniques established for the classical setting. Indeed, the possibility of direct transfer and use of results is missing and much has to be established afresh. However, recent developments seem to gradually fill-in the gap with the classical (distributive, or Boolean) setting. Other than the generalization of the concept of canonical (perfect) extensions to bounded lattices \cite{mai-harding} and the topological dualities of \cite{kata2z,duality2,choiceFreeHA,choiceFreeStLog}, a Sahlqvist theory for distribution-free logics in the RS-framework has been developed by Conradie and Palmigiano \cite{conradie-palmigiano}, a Goldblatt-Thomason theorem for non-distributive logics in a two-sorted semantic framework has been published by Goldblatt \cite{goldblatt-morphisms2019}, a van Benthem characterization result for logics without distribution has been detailed in \cite{vb} by the author and several distribution-free systems have been modeled in the two-sorted approach, amongst which we may mention \cite{mai-grishin,Almeida09,COUMANS201450,redm,choiceFreeStLog}.

This article fills in a gap in the literature by presenting a systematic study of the relational semantics of the basic normal modal logic systems that may lack distribution of conjunction over disjunction and conversely.  Relying on the extension of canonical extensions to the case of bounded lattices \cite{mai-harding}, it specifies what a canonical frame for the distribution-free systems is, it identifies frame classes corresponding to the basic normal modal logic systems and it proves canonicity and completeness in the relational semantics. 

Section~\ref{algebras and logics} defines normal lattice expansions and, in particular, the variety of implicative modal lattices. It specifies the basic proof system for any normal lattice expansion, built on top of a proof system for Positive Lattice Logic (\texttt{PLL}) and including only axioms and rules from which the distribution and normality properties of the logical operators are obtained. The section is concluded with a duality between the lattice of axiomatic extensions of the basic logic and the lattice of subvarieties of the variety determined by the equivalent algebraic semantics of the basic logic system. 

Section~\ref{frames introduction section} introduces sorted residuated frames with sorted relations and it defines their complete lattices of stable sets. Sorted image operators are defined as in the classical J\'{o}nsson-Tarski framework \cite{jt1} (which precedes by a decade and logically underlies Kripke semantics for modal logic, or Routley-Meyer semantics for relevance logic). Stable-set sorted operators are defined from the image operators and it is proven that, under a mild assumption of smoothness of the frame relations, they are completely additive. Normal lattice operators of distribution types depending on the sort of the generating relation are subsequently derived and the full complex algebra of a frame is defined. Relational semantics is presented in terms of both a relation of satisfaction $\forces$ and a relation of co-satisfaction (refutation) $\dforces$. The sorted relational framework is a smooth generalization of classical relational semantics (see Remark~\ref{jt rem1} and Remark~\ref{jt rem2}). The section is concluded with instantiating definitions to the case of implication and the modal operators and deriving soundness of the logic obtained as the fusion of its fragments with the single operators $\ra, \Box$ and $\Diamond$ and this completes the presentation of the basic background needed.

Section~\ref{minimal normal section} considers the minimal distribution-free normal modal logic system, including just the K-axiom $\Box(p\ra q)\proves\Box p\ra\Box q$, on top of the standard distribution and normality axioms for each of $\Box$ and $\Diamond$, but without the interdefinability of box and diamond axiom. To define frames that validate the K-axiom, it is first pointed out that the lattice $\gpsi$ of stable sets and the MacNeille completion of a sub-poset $Q\subseteq\gpsi$ of principal (closed and open) elements are identical. We refer to the sub-poset $Q$ as the intermediate structure of the frame. The significant issue arising is that in frames that can be axiomatically extended to validate the K-axiom the upper MacNeille extension $\lbb^u$ of the definition of the box operator $\lbb$ in the intermediate structure must lie below its lower MacNeille extension $\lbb^\ell$. A restricted (second-order) join-distributivity of $\lbb$ over joins, $\lbb\bigvee_{x\in A}\Gamma x\subseteq\bigvee_{x\in A}\lbb\Gamma x$ is postulated and the class of refined frames defined. The section is concluded with a soundness proof for the class of refined frames whose intermediate structure validates the K-axiom, a requirement that can be easily phrased as a first-order frame axiom.

Section~\ref{extensions section} specifies frame classes corresponding to the D, T, B, S4 and S5 axioms and proves soundness results. 

Section~\ref{canonical frame section} constructs the canonical frame, as an instance of a general representation result for normal lattice expansions. It is argued that the full complex algebra of the canonical frame is a canonical extension of the Lindenbaum-Tarski algebra of the logic, that the frame is refined and that the frame axioms corresponding to logic axioms hold. This establishes canonicity and completeness in the relational semantics.

In Section~\ref{distributive/Heyting section} we argue that the semantic framework and approach of the previous sections is scalable and it can be adapted to handle the cases where the underlying lattice logic is distributive, or intuitionistic. For the Heyting algebra case, we show that validity of the K-axiom holds without a need to assume a second-order axiomatization of the frame.

\section{Implicative Modal Lattices and Logics}
\label{algebras and logics}

\subsection{Implicative Modal Lattices}
Let $\{1,\partial\}$ be a 2-element set, $\mathbf{L}^1=\mathbf{L}$ and $\mathbf{L}^\partial=\mathbf{L}^\mathrm{op}$ (the opposite lattice). Extending established terminology \cite{jt1}, a function $f:\mathbf{L}_1\times\cdots\times\mathbf{L}_n\lra\mathbf{L}_{n+1}$ will be called {\em additive} and {\em normal}, or a {\em normal operator}, if it distributes over finite joins of the lattice $\mathbf{L}_i$, for each $i=1,\ldots n$, delivering a join in $\mathbf{L}_{n+1}$.

An $n$-ary operation $f$ on a bounded lattice $\mathbf L$ is {\em a normal lattice operator of distribution type  $\delta(f)=(j_1,\ldots,j_n;j_{n+1})\in\{1,\partial\}^{n+1}$}  if it is a normal additive function  $f:{\mathbf L}^{j_1}\times\cdots\times{\mathbf L}^{j_n}\lra{\mathbf L}^{j_{n+1}}$ (distributing over finite joins in each argument place), where  each $j_k$, for  $k=1,\ldots,n+1$,   is in the set $\{1,\partial\}$, hence ${\mathbf L}^{j_k}$ is either $\mathbf L$, or ${\mathbf L}^\partial$.

If $\tau$ is a tuple (sequence) of distribution types, a {\em normal lattice expansion (NLE) of (similarity) type $\tau$} is a lattice with a normal lattice operator of distribution type $\delta$ for each $\delta$ in $\tau$. By its definition, the class of normal lattice expansions of some type $\tau$ is equationally definable (a variety) $\mathbb{NLE}_\tau$.

\begin{defn} \label{implicative lattice defn}
An {\em implicative lattice}  is a bounded lattice with a binary (implication) operation,  $\mathbf{L}=(L,\leq,\wedge,\vee,0,1,\ra)$, of distribution type $\delta(\ra)=(1,\partial;\partial)$, i.e. the following axioms hold, in addition to the axioms for bounded lattices.
\begin{tabbing}
\hskip5mm\=(A1)\hskip5mm\= $(a\vee b)\ra c=(a\ra c)\wedge(b\ra c)$\\
\>(A2)\> $a\ra(b\wedge c)=(a\ra b)\wedge(a\ra c)$\\
\>(N)\> $(0\ra a) = 1 = (a\ra 1)$
\end{tabbing}
The lattice is  {\em distributive} if axiom (A3) is assumed.
\begin{tabbing}
\hskip5mm\=(A3)\hskip5mm\= $a\wedge(b\vee c)=(a\wedge b)\vee(a\wedge c)$
\end{tabbing}
It is a Heyting algebra if axioms (H1)-(H2) are added to the axioms for bounded lattices
\begin{tabbing}
\hskip5mm\=(H1)\hskip5mm\= $a\wedge(a\ra b)\leq b$\\
\>(H2)\> $b\leq a\ra(a\wedge b)$
\end{tabbing}
An {\em implicative modal lattice}  $\mathbf{L}=(L,\leq,\wedge,\vee,0,1,\ra,\Box,\Diamond)$ is an implicative lattice with two unary (modal) operators $\Box,\Diamond$ satisfying the axioms
\begin{tabbing}
\hskip5mm\=(M$\Box$)\hskip5mm\= $\Box(a\wedge b)=\Box a\wedge\Box b$\hskip1cm\=(M$\Diamond$)\hskip5mm\=$\Diamond(a\vee b)=\Diamond a\vee\Diamond b$\\
\>(N$\Box$) \> $\Box 1=1$\>(N$\Diamond$) \> $\Diamond 0=0$
\end{tabbing}
\end{defn}

The weak modal logics we will consider arise by adding to the axiomatization any of the K, D, T, B, or S4, S5 axioms, as in the standard systems of normal modal logic, except for dropping the assumption that the underlying propositional logic is classical and replacing it by the weaker assumption that it is the logic of an implicative modal lattice.
\begin{tabbing}
\hskip5mm\=(K)\hskip5mm\= $\Box(a\ra b)\leq\Box a\ra\Box b$\hskip1cm\=
   (B)\hskip8mm\=  $a\leq{\Box\Diamond} a$ and ${\Diamond\Box}a\leq a$\\
\>(D)\> $\Box a\leq \Diamond a$ \>(S4$_\Box$)\> $\Box a\leq{\Box\Box} a$\\
\>(T$_\Box$)\> $\Box a\leq a$ \>(S4$_\Diamond$)\> ${\Diamond\Diamond}a\leq\Diamond a$\\
\>(T$_\Diamond$)\> $a\leq \Diamond a$ \>(S5)\> $\Diamond a\leq{\Box\Diamond}a$ and ${\Diamond\Box} a\leq{\Box} a$
\end{tabbing}

To fix an understanding of what constitutes a canonical frame and model for our logics, we review some basic definitions on completions and canonical extensions of lattices.

A {\em completion} $(\alpha,\mathbf{C})$ of an NLE $\mathbf{L}$ is a complete NLE $\mathbf{C}$ of the same similarity type with an embedding $\alpha:\mathbf{L}\lra\mathbf{C}$. In \cite{jt1} a notion of canonical (perfect) extension of Boolean algebras was identified, extended by Gehrke and J\'{o}nsson \cite{mai-jons} to the case of distributive lattices and by Gehrke and Harding \cite{mai-harding} to the case of mere bounded lattices. For the purposes of this article, the completions we work with are canonical extensions, traditionally associated to a canonical frame construction in completeness proofs. 

A canonical extension \cite{mai-harding} of a lattice $\mathbf{L}$ is a completion $(\alpha,\mathbf{C})$ of the lattice such that the following density and compactness requirements are satisfied
\begin{quote}
$\bullet$ (density) $\alpha[{\mathbf L}]$ is {\em dense} in $\mathbf{C}$, where the latter means that every element of $\mathbf{C}$ can be expressed both as a meet of joins and as a join of meets of elements in $\alpha[{\mathbf L}]$\\[0.5mm]
$\bullet$ (compactness) for any set $A$ of closed elements and any set  $B$ of open elements of $C$, $\bigwedge A\leq\bigvee B$ iff there exist finite subcollections $A_1\subseteq A, B_1\subseteq B$ such that $\bigwedge A_1\leq\bigvee B_1$
\end{quote}
where the {\em closed elements} of $\mathbf{C}$ are defined in \cite{mai-harding} as the elements in the meet-closure of the embedding map $\alpha$ and the {\em open elements} of $\mathbf{C}$ are defined dually as the join-closure of the image of $\alpha$. Canonical extensions are unique up to an isomorphism that commutes with the embeddings \cite[Proposition~2.7]{mai-harding}.

For a unary lattice operator $f:\mathbf{L}\lra\mathbf{L}$,  its $\sigma$ and $\pi$-extension in a canonical extension $\mathbf{C}$ of the lattice $\mathbf{L}$ is defined in \cite[Definition~4.1, Lemma~4.3]{mai-harding} by equation \eqref{sigma/pi-extn}, where $\type{K}$ is the set of closed elements of $\mathbf{C}$ and $\type{O}$ is its set of open elements
\begin{equation}\label{sigma/pi-extn}
\begin{array}{lcl}
f^\sigma(k)=\bigwedge\{f(a)\midsp k\leq a\in L\} &\hskip1cm& f_\sigma(u)=\bigvee\{f^\sigma(k)\midsp\type{K}\ni k\leq u\}\\
f^\pi(o)=\bigvee \{f(a)\midsp L\ni a\leq o\} && f_\pi(u)=\bigwedge\{f^\pi(o)\midsp u\leq o\in\type{O}\}
\end{array}
\end{equation}
where in these definitions $\mathbf{L}$ is identified with its isomorphic image in $\mathbf{C}$ and $a, f(a)\in {\cal L}$ are identified with their representation images.

\begin{defn}\label{canonical extn of NLEs}
A {\em canonical extension of a normal lattice expansion} $\mathbf{L}=(L,\leq,\wedge,\vee,0,1,(f_j)_{j\in J})$ is a canonical extension $(\alpha,\mathbf{L}^\sigma)$ of the underlying bounded lattice, with normal lattice operators $F_j$ extending $f_j$, for each $j\in J$, and such that $F_j=f^\sigma_j$, whenever $\delta_j(n+1)=1$ (the operator $f_j$ returns joins) and $F_j=f^\pi_j$, whenever $\delta_j(n+1)=\partial$ (the operator $f_j$ returns meets).

A {\em variety of normal lattice expansions is canonical} if it is closed under canonical extensions of its members. A logic is canonical if the variety of its equivalent algebraic semantics is canonical. Canonicity implies completeness in relational semantics, if the full complex algebra of the canonical frame for the logic is a canonical extension of the Lindenbaum-Tarski algebra of the logic.
\end{defn}

\subsection{Languages, Logics and Algebraic Semantics}
\label{logis and algebraic semantics}
Let $J$ be a countable set and $\delta:J\lra\{1,\partial\}^{n(j)+1}$ a map. The propositional language $\mathcal{L}=\mathcal{L}(J,\delta)$ of normal lattice expansions of similarity type $\tau=(\delta_j)_{j\in J}$ is defined by the grammar
\begin{eqnarray}
\mathcal{L}\ni\varphi &=& p_i(i\in \mathbb{N})\midsp\top\midsp\bot\midsp \varphi\wedge\varphi\midsp\varphi\vee\varphi\midsp (f_j(\varphi_1,\ldots,\varphi_{n(j)}))_{j\in J}.\label{language}
\end{eqnarray}
 Where $\delta_j=(j_1,\ldots,j_{n(j)};j_{n(j)+1})$ we let $\delta_j(k)=j_k$ for $1\leq k\leq j_{n(j)+1}$. 
 The  language $\mathcal{L}^\ra_{\Box\Diamond}$ is an instance of the language $\mathcal{L}_\tau(J,\delta)$ of NLEs of an arbitrary type  $\tau$, where in the current context we may let the index set be $J=\{\Diamond,\Box,\ra\}$ and $f_\Diamond=\dd,f_\Box=\bb, f_\ra={\rfspoon}$  (and with $\delta_\Diamond=(1;1), \delta_\Box=(\partial,\partial)$ and  $\delta_\ra=(1,\partial;\partial)$), displayed below 
 \begin{eqnarray}
\mathcal{L}^\ra_{\Box\Diamond}\ni\varphi &=& p_i(i\in \mathbb{N})\midsp\ufootl\midsp\dfootl\midsp \varphi\wedge\varphi\midsp\varphi\vee\varphi\midsp \bb\varphi\midsp\dd\varphi\midsp\varphi\rfspoon\varphi.\label{modal language}
\end{eqnarray}

For a fixed $\tau$ and an algebra $\mathbf{A}$ (a normal lattice expansion of type $\tau$) a model on $\mathbf{A}$ is defined as usual as a pair $\mathfrak{M}=(\mathbf{A},v)$, given an interpretation function $v:P\lra A$, where $A$ is the underlying set of $\mathbf{A}$, and letting also $v(\bot)=0, v(\top)=1$. The interpretation $\val{\;}_\mathfrak{M}$ of all sentences is obtained as usual, as the unique homomorphic extension of $v$ from the absolutely free term-algebra on $\mathcal{L}$ to $\mathbf{A}$. Treating $\mathcal{L}$ also as a term language for $\mathbf{A}$, a formal equation is a pair of sentences (terms, designating lattice elements), written $\varphi\approx\psi$. A formal equation is true in a model $\mathfrak{M}=(\mathbf{A},v)$, written $\mathfrak{M}\models\varphi\approx\psi$ iff $\val{\varphi}_\mathfrak{M}=\val{\psi}_\mathfrak{M}$. An algebra $\mathbf{A}$ validates a formal equation, $\mathbf{A}\models\varphi\approx\psi$, iff for any model $\mathfrak{M}=(\mathbf{A},v)$ on $\mathbf{A}$ we have $\mathfrak{M}\models\varphi\approx\psi$. Similarly for a class $\mathbb{C}$ of algebras. A formal inequation is syntactically defined as usual, by $(\varphi\preceq\psi) = (\varphi\approx\varphi\wedge\psi)$. 

For our purposes, we define a logic on the language $\mathcal{L}_\tau$ as a set of pairs $(\varphi,\psi)$, which we write as $\varphi\proves\psi$ and refer to as `sequents'. A proof system for the logic is a set of initial sequents (axioms) and rules {\small $\infrule{S_1\cdots S_n}{S}$}, where $S_i,S$ are sequents. Proofs (in a proof system) are defined as usual. The logic generated by a proof system is the set of the sequents it proves. A sequent is valid in a model, written $\mathfrak{M}\forces\varphi\proves\psi$ (or $\varphi\forces_\mathfrak{M}\psi$, omitting the subscript if clear from context) iff $\val{\varphi}_\mathfrak{M}\leq\val{\psi}_\mathfrak{M}$, which is exactly the condition for validity of inequations $\mathfrak{M}\models\varphi\preceq\psi$. Analogously for validity in an algebra or in a class of algebras. A proof system is sound in a class of algebras iff every provable sequent (theorem) $\varphi\proves\psi$ is valid in this class iff the initial sequents are valid and every rule is valid in this class.

The minimal logic $\mathbf{\Lambda}_\tau$ (for any similarity type $\tau$), whose proof system we define in Table~\ref{minimal proof system}, is defined as an extension of a proof system for Positive Lattice Logic $\texttt{PLL}$, the logic of bounded lattices. To simplify, we use vectorial notation $\vec{\varphi}$ for a tuple of sentences, we let $\vec{\varphi}[\;]_k$ be the vector with an empty place at the $k$-th argument position and $\vec{\varphi}[\psi]_k$ either to display the sentence at the $k$-position, or to designate the result of filling the empty place in $\vec{\varphi}[\;]_k$, or the result of replacing the entry $\varphi_k$ with $\psi$ at the $k$-th argument place. 

\begin{table}[t]
\caption{Proof System for the Minimal Logic $\mathbf{\Lambda}_\tau$}
\label{minimal proof system}
\mbox{}\\[5mm]
\begin{tabular}{llllll}
\multicolumn{3}{l}{\underline{{\bf PLL} axioms and rules}}\\[3mm]
 $\varphi\proves\varphi$ & $\infrule{\varphi\proves\vartheta}{\varphi\wedge\psi\proves\vartheta}$  & $\infrule{\psi\proves\vartheta}{\varphi\wedge\psi\proves\vartheta}$  &
 $\infrule{\varphi\proves\psi\hskip5mm\varphi\proves\vartheta}{\varphi\proves\psi\wedge\vartheta}$  &  $\infrule{\varphi\proves\psi\hskip5mm\psi\proves\vartheta}{\varphi\proves\vartheta}$(Cut)  \\[3mm]
 $\bot\proves\varphi$ & $\varphi\proves\top$ &
  $\infrule{\vartheta\proves\varphi}{\vartheta\proves\varphi\vee\psi}$ & $\infrule{\vartheta\proves\psi}{\vartheta\proves\varphi\vee\psi}$ & $\infrule{\varphi\proves\vartheta\hskip5mm\psi\proves\vartheta}{\varphi\vee\psi\proves\vartheta}$\\[3mm]
 \multicolumn{3}{l}{Substitution rule} & $\infrule{\varphi(p)\proves\psi(p)}{\varphi[\vartheta/p]\proves\psi[\vartheta/p]}$
  \\[5mm]
\multicolumn{3}{l}{\underline{Additional axioms/rules for $\mathbf{\Lambda}_\tau$}}\\[3mm]
  \multicolumn{5}{l}{
 For each $j\in J$ and where $\delta_j=(j_1,\ldots,j_k,\ldots, j_{n(j)};j_{n(j)+1})$ is the distribution type of $f_j$}\\[2mm]
 \multicolumn{3}{l}{$\bullet$ Monotonicity/Antitonicity rules}\\[2mm]
   \multicolumn{3}{l}{
  $\infrule{\psi\proves\vartheta}{f_j(\vec{\varphi}[\psi]_k)\proves f_j(\vec{\varphi}[\vartheta]_k)}$ if $j_k=j_{n(j)+1}$
  }
  &
  \multicolumn{3}{l}{
  $\infrule{\psi\proves\vartheta}{f_j(\vec{\varphi}[\vartheta]_k)\proves f_j(\vec{\varphi}[\psi]_k)}$ if $j_k\neq j_{n(j)+1}$
  }\\[4mm]
\multicolumn{3}{l}{$\bullet$  Distribution axioms}\\[2mm]
 \multicolumn{4}{l}{
 $f_j(\vec{\varphi}[\psi\vee\chi]_k)\proves f_j(\vec{\varphi}[\psi]_k)\vee f_j(\vec{\varphi}[\chi]_k)$\hskip1cm if $j_k=1=j_{n(j)+1}$
 }\\[3mm]
  \multicolumn{4}{l}{
 $f_j(\vec{\varphi}[\psi]_k)\wedge f_j(\vec{\varphi}[\chi]_k)\proves f_j(\vec{\varphi}[\psi\wedge\chi]_k)$ \hskip1cm if $j_k=\partial=j_{n(j)+1}$
 }\\[4mm]
 \multicolumn{3}{l}{$\bullet$  Co-Distribution axioms}\\[2mm]
 \multicolumn{4}{l}{
$f_j(\vec{\varphi}[\psi\vee\chi]_k)\proves f_j(\vec{\varphi}[\psi]_k)\wedge f_j(\vec{\varphi}[\chi]_k)$\hskip1cm   if $j_k=1\neq\partial=j_{n(j)+1}$}\\[3mm]
 \multicolumn{4}{l}{
 $f_j(\vec{\varphi}[\psi]_k)\vee f_j(\vec{\varphi}[\chi]_k)\proves  f_j(\vec{\varphi}[\psi\wedge\chi]_k)$
  \hskip1cm if $j_k=\partial\neq 1=j_{n(j)+1}$}
  \\[5mm]
\multicolumn{3}{l}{$\bullet$  Normality axioms}\\[2mm] 
 \multicolumn{4}{l}{
$f_j(\vec{\varphi}[\bot]_k)\proves\bot$ \hskip1cm if $j_k=1=j_{n(j)+1}$}\\[3mm]
 \multicolumn{4}{l}{
$\top\proves f_j(\vec{\varphi}[\top]_k)$ \hskip1cm if $j_k=\partial=j_{n(j)+1}$}\\[3mm]
 \multicolumn{4}{l}{
 $\top\proves f_j(\vec{\varphi}[\bot]_k)$ \hskip1cm if  $j_k=1\neq\partial=j_{n(j)+1}$}\\[3mm]
 \multicolumn{4}{l}{
 $f_j(\vec{\varphi}[\top]_k)\proves\bot$      \hskip1cm   if $j_k=\partial\neq 1=j_{n(j)+1}$}\\[3mm]
\end{tabular}

\vspace*{2mm}
\hrulefill
\end{table}

Specializing rules for the case of the logical operators $\Box,\Diamond$ and $\ra$ can be safely left to the reader.

Writing $\varphi\equiv\psi$ to mean that both $\varphi\proves\psi$ and $\psi\proves\varphi$ are provable sequents and $[\varphi]$ for the equivalence class of $\varphi$, the axioms and rules of the minimal logic $\mathbf{\Lambda}_\tau$ ensure that $\equiv$ is a congruence and that the Lindenbaum-Tarski algebra of the logic is a normal lattice expansion of type $\tau=(\delta_j)_{j\in J}$. This implies a completeness theorem, stated below.
\begin{thm}
\label{minimal algebraic completeness}
For any similarity type $\tau$ the minimal logic $\mathbf{\Lambda}_\tau$ is sound and complete in the variety $\mathbb{NLE}_\tau$ of normal lattice expansions of similarity type $\tau$.\telos
\end{thm}

If $\Gamma$ is a set of sequents, then $\mathbf{\Lambda}_\tau+\Gamma$ stands for the logic that results by adding members of $\Gamma$ as new initial sequents (notice that we have included an explicit substitution rule in the proof system). Denote by $\texttt{Ext}\mathbf{\Lambda}_\tau$ the family of such extensions of $\mathbf{\Lambda}_\tau$ and observe that \texttt{Ext}$\mathbf{\Lambda}_\tau$ is a complete lattice under intersection.

Every extension $\mathbf{\Lambda}=\mathbf{\Lambda}_\tau+\Gamma$ determines a subvariety $\mathbb{V}\in\texttt{Sub}\mathbb{NLE}_\tau$ whose equational theory is $\mathcal{E}(\mathbb{V})=\mathcal{E}(\mathbb{NLE}_\tau)\cup\{\varphi\preceq\psi\midsp (\varphi,\psi)\in\Gamma\}$.

\begin{defn}\rm
The maps $\mathrm{V}:\texttt{Ext}\mathbf{\Lambda}_\tau\leftrightarrows\texttt{Sub}\mathbb{NLE}_\tau:\Lambda$ are defined on a logic $\mathbf{\Lambda}\in\texttt{Ext}\mathbf{\Lambda}_\tau$ and a variety $\mathbb{V}\in\texttt{Sub}\mathbb{NLE}_\tau$ by
\begin{align}
\mathrm{V}(\mathbf{\Lambda})&=\{\mathbf{L}\in\mathbb{NLE}_\tau\midsp\forall(\varphi,\psi)\in\mathbf{\Lambda}\;\varphi\forces_\mathbf{L}\psi\}\\
\Lambda(\mathbb{V})&=\{(\varphi,\psi)\midsp\varphi\forces_\mathbb{V}\psi\}
\end{align}
\end{defn}

The following Propositions have rather easy proofs, left to the interested reader. Analogous proofs can be found in \cite[Propositions~2.7--2.9]{choiceFreeStLog}.
\begin{prop}[Definability]\rm
\label{definability}
For any subvariety $\mathbb{V}\in\texttt{Sub}\mathbb{NLE}_\tau$ and any normal lattice expansion $\mathbf{L}\in\mathbb{NLE}_\tau$, $\mathbf{L}\in\mathbb{V}$ iff $\mathbf{L}\forces\Lambda(\mathbb{V})$. Equivalently, $\mathrm{V}\Lambda(\mathbb{V})=\mathbb{V}$.\telos
\end{prop}

\begin{prop}[Completeness]\rm
\label{completeness}
For every logic $\mathbf{\Lambda}\in\texttt{Ext}\mathbf{\Lambda}_\tau$ and every pair of sentences $(\varphi,\psi)$, $\varphi\proves\psi\in\mathbf{\Lambda}$ iff $\varphi\forces_{\mathrm{V}(\mathbf{\Lambda})}\psi$. Equivalently, $\Lambda\mathrm{V} (\mathbf{\Lambda})=\mathbf{\Lambda}$.\telos
\end{prop}

\begin{prop}[Duality]\rm
\label{duality}
The maps $\mathrm{V},\Lambda$ constitute a complete lattice dual isomorphism $\mathrm{V}:\texttt{Ext}\mathbf{\Lambda}_\tau\iso\texttt{Sub}\mathbb{NLE}_\tau^\mathrm{op}:\Lambda$.\telos
\end{prop}

This dual isomorphism allows us to switch from logics to varieties of normal lattice expansions and back without any loss of information.

\section{Frames and Relational Semantics}
\label{frames introduction section}
\subsection{Frames and Models}
By a (relational) frame we mean a structure $\mathfrak{F}=(s,Z,I,(R_j)_{j\in J},\sigma)$, where $s$ is a list of sorts, $Z=(Z_t)_{t\in s}$ is a nonempty sorted set (i.e. none of the sorts $Z_t$ is allowed to be empty), where we make no assumption of disjointness of sorts, $I\subseteq\prod_{t\in s}Z_t$ is a distinguished sorted relation, $\sigma$ is a sorting map on $J$ with  $\sigma(j)\in s^{n(j)+1}$ and $(R_j)_{j\in J}$ is a family of sorted relations such that if $\sigma(j)=(j_{n(j)+1};j_1,\ldots,j_{n(j)})$, then $R_j\subseteq Z_{j_{n(j)+1}}\times\prod_{k=1}^{n(j)}Z_{j_k}$. The {\em sort $\sigma(R_j)$ (or just $\sigma(j)$, or $\sigma_j$) of the relation} $R_j$ is the tuple $\sigma(j)=(j_{n(j)+1};j_1\cdots j_{n(j)})$. The similarity type of the structure $\mathfrak{F}=(s,Z,I,(R_j)_{j\in J},\sigma)$ is the tuple $\langle\sigma(j)\rangle_{j\in J}$ of the sorts of the relations $R_j$ in the structure.

We will work with structures where
$s=\{1,\partial\}$, $Z=(Z_1,Z_\partial)$ is a sorted set and $I\subseteq Z_1\times Z_\partial$ is a distinguished sorted relation. We often display the sort of a relation as a superscript, as in $\mathfrak{F}=(s,Z,I,(R^{\sigma(j)}_j)_{j\in J})$. For example, $R^{11}, T^{\partial 1\partial}$ designate sorted relations $R\subseteq Z_1\times Z_1$ and $T\subseteq Z_\partial\times Z_1\times Z_\partial$. In the intended application of the present article the frame relations considered are $R^{11}_\Diamond, R^{\partial\partial}_\Box$ and $R^{\partial 1\partial}_\ra$, but we shall use $T$ for the latter (or $T^{\partial 1\partial}$, displaying its sort), as this makes it easier to relate to results obtained in \cite{choiceFreeHA}.

The relation $I$ generates a residuated pair $\lambda:\powerset(Z_1)\leftrightarrows\powerset(Z_\partial):\rho$, defined as usual by
\[
\lambda(U)=\{y\in Z_\partial\midsp\exists x\in Z_1(xIy\wedge x\in U)\}\hskip2mm
\rho(V)=\{x\in Z_1\midsp\forall y\in Z_\partial(xIy\lra y\in V)\}.
\]
We may also use the notation $\largediamond U$ for $\lambda U$ and $\lbbox V$ for $\rho V$, as we have often done in previous published work. 
The complement of $I$ will be designated by $\upv$ and we refer to it as the {\em Galois relation of the frame}. It generates a Galois connection $(\;)\rperp:\powerset(Z_1)\leftrightarrows\powerset(Z_\partial)^{\rm op}:\rperp(\;)$ defined by
\[
U\rperp=\{y\in Z_\partial\midsp\forall u\in Z_1 (u\in U\lra u\upv y)\}\hskip2mm
\rperp V=\{x\in Z_1\midsp\forall y\in Z_\partial(y\in V\lra x\upv y)\}.
\]

Observe that the closure operators generated by the residuated pair and the Galois connection are identical, i.e. $\rho\lambda U=\rperp(U\rperp)$ and $\lambda\rho V=(\rperp V)\rperp$. This follows from the fact that $U\rperp=\largesquare(-U)$ and ${}\rperp V=\lbbox(-V)$. 

To simplify, we often use a priming notation for both Galois maps $(\;)\rperp$ and $\rperp(\;)$, i.e. we let $U'=U\rperp$, for $U\subseteq Z_1$, and $V'=\rperp V$, for $V\subseteq Z_\partial$. Hence $U''=\rperp(U\rperp)=\rho\lambda U$ and $V''=(\rperp V)\rperp=\lambda\rho V$. 

The complete lattice of all {\em Galois stable} sets $Z_1\supseteq U=U''$ will be designated by $\mathcal{G}(Z_1)$ and the complete lattice of all {\em Galois co-stable} sets $Z_\partial\supseteq V=V''$ will be similarly denoted by $\mathcal{G}(Z_\partial)$. We refer to Galois stable and co-stable sets as {\em Galois sets}. Note that each of $Z_1, Z_\partial$ is a Galois set, but the empty set need not be Galois. Hence the bottom element in each of $\gpsi,\gphi$ is the closure of the empty set, $\emptyset''$.

For an element $u$ in either $Z_1$ or $Z_\partial$ and a subset $W$, respectively of $Z_\partial$ or $Z_1$, we write $u|W$, under a well-sorting assumption, to stand for either $u\upv W$ (which stands for $u\upv w$, for all $w\in W$), or $W\upv u$ (which stands for $w\upv u$, for all $w\in W$), where well-sorting means that either $u\in Z_1, W\subseteq Z_\partial$, or $W\subseteq Z_1$ and $u\in Z_\partial$, respectively. Similarly for the notation $u|v$, where $u,v$ are elements of different sort.

A preorder relation is defined on each of $Z_1,Z_\partial$ by $u\preceq w$ iff $\{u\}'\subseteq\{w\}'$. We call a frame {\em separated} if $\preceq$ is in fact a partial order $\leq$. For an element $u$ (of either $Z_1$ or $Z_\partial$) we write $\Gamma u$ for the set of elements $\preceq$-above it. We hereafter assume that frames are separated.

Sets $\Gamma w$ and $\{w\}'$ will be referred to as {\em principal elements}.  $\Gamma w$ will be referred to as a {\em closed element} and $\{w\}'$ as an {\em open element}. $\texttt{K}\mathcal{G}(Z_1), \texttt{O}\mathcal{G}(Z_1)$ designate the families of closed and open elements of $\gpsi$, and analogously for $\gphi$. A closed element $\Gamma x\in\gpsi$ is clopen if there exists an open element $\{y\}'\in\gpsi$ (which is unique, when it exists, in separated frames) such that $\Gamma x=\{y\}'$. We set $\texttt{KO}\mathcal{G}(Z_1)$ for the family of clopen elements. Similarly for $\gphi$.

A point $w$ such that $\Gamma w$ is a clopen element will be also referred to as a {\em clopen point}.

The following basic facts will be often used without reference to  Lemma~\ref{basic facts}.
\begin{lemma}
\label{basic facts}
Let $\mathfrak{F}=(s,Z,I,(R_j)_{j\in J},\sigma)$ be a  frame, $u\in Z=Z_1\cup Z_\partial$ and $\upv$ the Galois relation of the frame. Let $v|G$ refer to either $G\upv v$, if $G\in\gpsi, v\in Z_\partial$, or $v\upv G$, if $v\in G_1$ and $G\in\gphi$.
\begin{enumerate}
\item $\upv$ is increasing in each argument place (and thereby its complement $I$ is decreasing in each argument place).
\item $(\Gamma u)'=\{u\}'$ and $\Gamma u=\{u\}^{\prime\prime}$ is a Galois set.
\item Galois sets are increasing, i.e. $u\in G$ implies $\Gamma u\subseteq G$.
\item For a Galois set $G$, $G=\bigcup_{u\in G}\Gamma u$.
\item For a Galois set $G$, $G=\bigvee_{u\in G}\Gamma u=\bigcap_{v|G}\{v\}'$.
\item For a Galois set $G$ and any set $W$, $W^{\prime\prime}\subseteq G$ iff $W\subseteq G$.
\end{enumerate}
\end{lemma}
\begin{proof}
  By simple calculation. Proof details are included in \cite[Lemma 2.2]{sdl-exp}.  For claim 4, $\bigcup_{u\in G}\Gamma u\subseteq G$ by claim 3 (Galois sets are upsets). 
\end{proof}

\begin{defn}
\label{Galois dual relation}
For a sorted $(n+1)$-ary frame relation $R_j$, its {\em Galois dual relation} $R^\prime_j$ is defined by $R_j^\prime u_1\cdots u_n=(R_ju_1\cdots u_n)'$.
\end{defn} 

Notation is simplified by using vectors $\vec{u}=u_1\cdots u_n$, so that the definition is $R_j^\prime\vec{u}=(R_j\vec{u})'$. We let $\vec{u}[\;]_k$ be the vector with a hole (or just a place-holder) at the $k$-th position and write $u[w]_k$ either to display the element at the $k$-th place, or to designate the result of filling the $k$-th place of $u[\;]_k$, or to denote the result of replacing the element $u_k$ in $\vec{u}$ by the element $w$.
For $1\leq k\leq n$, the {\em $k$-th section of an $(n+1)$-ary relation $S$} is the set $wS\vec{u}[\;]_k$. For $k=n+1$ the section is simply the set $S\vec{u}$. 

\begin{defn}
\label{smooth defn}
Call a frame relation $R_j$ {\em smooth} iff every section of its Galois dual relation $R_j^\prime$ is a Galois set (stable, or co-stable, according to the sort $\sigma(R_j)$ of the relation). 
\end{defn}

Hereafter, when considering a structure $\mathfrak{F}=(s,Z,I,(R_j)_{j\in J},\sigma)$ we always assume that the frame is separated and that all frame relations are smooth.
Since no other kind of frame (relational structure) will be considered in this article, we shall refer to relational structures $\mathfrak{F}=(s,Z,I,(R_j)_{j\in J},\sigma)$ simply as frames, or sorted residuated frames.

\begin{rem}[Classical Kripke Frames]\label{jt rem1}
Structures $\mathfrak{F}=(s,Z,I,(R_j)_{j\in J},\sigma)$, as described above, generalize classical Kripke frames for (poly)modal logic with polyadic modalities, arising by letting  $Z_1=Z_\partial$ and where $I\subseteq Z_1\times Z_\partial$ is the identity relation.
  
For a frame relation $R_j$, for example a binary relation $R_\Diamond$, its Galois dual relation is defined (when $\mathfrak{F}$ is a classical Kripke frame) by $yR'_\Diamond x$ iff for all $z$, if $zR_\Diamond x$, then $z\neq y$, so that the section $R'_\Diamond x$ is the complement $-R_\Diamond x$ of the section $R_\Diamond x$. The smoothness requirement for frame relations is trivially satisfied, since all sets are stable, $--U=U$.

Furthermore, the preorder $x\preceq z$ iff $\{x\}'\subseteq\{z\}'$ iff (in a classical Kripke frame) $-\{x\}\subseteq -\{z\}$ is  the identity relation, so that $\Gamma x=\{x\}$ and ${}\rperp\{y\}=-\{y\}$. In other words, the principal elements (closed and open) in the frame are precisely the atoms and co-atoms of a powerset Boolean algebra. 
\end{rem}

A relational model $\mathfrak{M}=(\mathfrak{F},V)$ consists of a frame $\mathfrak{F}$ and a sorted valuation $V=(V^1,V^\partial)$ of propositional variables, interpreting a variable $p$ as a Galois stable set $V^1(p)\in\gpsi$ and co-interpreting it as a Galois co-stable set $V^\partial(p)=V^1(p)\rperp\in\gphi$.  Interpretations and co-interpretations determine each other in the sense that for any sentence $\varphi\in\mathcal{L}_\tau$, if $\val{\varphi}\in\mathcal{G}(Z_1)$ is an interpretation extending a valuation $V^1$ of propositional variables as stable sets, then $\val{\varphi}\rperp=\yvval{\varphi}\in\mathcal{G}(Z_\partial)$  is the co-interpretation extending the valuation $V^\partial$.

Satisfaction ${\forces}\subseteq Z_1\times\mathcal{L}_\tau$ and co-satisfaction (refutation) ${\dforces}\subseteq Z_\partial\times\mathcal{L}_\tau$ relations are then defined as expected, by $Z_1\ni x\forces\varphi$ iff $x\in\val{\varphi}$ and $Z_\partial\ni y\dforces\varphi$ iff $y\in\yvval{\varphi}$. Since satisfaction and co-satisfaction determine each other, for each operator it suffices to provide either its satisfaction, or its co-satisfaction (refutation) clause, in line with the principle of order-dual relational semantics introduced in \cite{odigpl}, as we do in Table \ref{sat}. The relation $R^{111}$ that appears in the satisfaction clause for implication is defined in the definition below.
\begin{defn}
\label{derived relations defn}
Define the relation $R^{111}\subseteq Z_1\times(Z_1\times Z_1)$ from the frame relation $T$ as follows:
\begin{tabbing}
\hskip2mm\=$T^{11\partial}$ \hskip3mm\= Galois dual relation of $T^{\partial 1\partial}$ \hskip3mm\= $xT^{11\partial}zv$ iff $\forall y\in Z_\partial(yT^{\partial 1\partial}zv\lra x\upv y)$\\[1mm]
\>$R^{\partial 11}$ \> argument permutation \> $vR^{\partial 11}zx$ iff $xT^{11\partial}zv$\\
\>$R^{111}$ \> Galois dual relation of $R^{\partial 11}$ \> $uR^{111}zx$ iff $\forall v\in Z_\partial(vR^{\partial 11}zx\lra u\upv v)$.
\end{tabbing}
\end{defn} 

\begin{table}[t]
\caption{(Co)Satisfaction relations}
\label{sat}
\begin{tabbing}
$x\forces p_i$\hskip8mm\=iff\hskip3mm\= $x\in V^1(p_i)$\\
$x\forces\ufootl$ \>iff\> $x=x$\\
$y\dforces\dfootl$\>iff\> $y=y$\\
$x\forces\varphi\wedge\psi$\>iff\> $x\forces\varphi$ and $x\forces\psi$\\
$y\dforces\varphi\vee\psi$\>iff\> $y\dforces\varphi$ and $y\dforces\psi$\\[2mm]
$y\dforces \dd\varphi$ \hskip6mm\>iff\> $\forall v\in Z_\partial\;(yR_\Diamond''v\lra v\dforces\varphi)$\\
$x\forces \bb\varphi$ \>iff\> $\forall z\in Z_1\;(xR_\Box''z\lra z\forces\varphi)$\\
$x\forces\varphi\rfspoon\psi$\>iff\> $\forall u\in Z_1\forall y\in Z_\partial(u\forces\varphi\;\wedge\; y\dforces\psi\lra xT'uy)$\\
\>iff\> $\forall u,z\in Z_1(u\forces\varphi\wedge zR^{111}ux\lra z\forces\psi)$
\end{tabbing}
\hrule
\end{table}
\noindent
In the satisfaction clause for implication, $T'$ is the Galois dual relation of $T$. In the clauses for the modal operators, we define the double dual $yR_\Diamond''=\lperp(yR_\Diamond')$ (and recall that the Galois dual $R_\Diamond'$ is defined from $R_\Diamond\subseteq Z_1\times Z_1$ by setting $R_\Diamond'z=(R_\Diamond z)\rperp$) and similarly for the double dual $R_\Box''$, defined from $R_\Box\subseteq Z_\partial\times Z_\partial$ by first letting $R_\Box'y=\lperp(R_\Box y)$ be the Galois dual relation, then defining $xR_\Box''=\lperp(xR_\Box')$.

There are two alternative but equivalent semantic clauses for implication, the first using the Galois dual $T'$ of $T^{\partial 1\partial}$, while the second uses a ternary relation $R^{111}$ on $Z_1$, derived from $T$. Equivalence of the two clauses is shown using  \cite[Proposition~3.6]{choiceFreeHA} and Proposition~\ref{Ra long and short}. The first clause is more familiar in a non-distributive setting, see for example \cite{suzuki-morphisms,mai-gen}. The second clause is familiar from the relational semantics of the implication connective of Relevance Logic \cite{relational,relevance2,relevance3}.

\subsection{Full Complex Algebras}
\label{full complex algebras section}
The logic $\mathbf{\Lambda}_{\Box\Diamond}^\ra$ is just the fusion $\mathbf{\Lambda}_{\Box\Diamond}^\ra=\mathbf{\Lambda}_\Box\oplus\mathbf{\Lambda}_\Diamond\oplus\mathbf{\Lambda}^\ra$ of its single-operator sub-systems, axiomatized only by the distribution and normality axioms (on top of the axiomatization for positive lattice logic) as in Table~\ref{minimal proof system}, hence its soundness follows from the soundness of its constituent systems. The proofs are given in Section~\ref{soundness for minimal},  Proposition~\ref{soundness for implication}, Proposition~\ref{soundness for box} and Proposition~\ref{soundness for diamond}, one for each subsystem. In all three Propositions, the argument proceeds by using the frame relations to define operations in the complete lattice $\gpsi$ of stable sets. This process is uniform, it relies on the duality for normal lattice expansions given in \cite{duality2} and it results in defining the dual full complex algebra $\mathfrak{F}^+$ of a frame $\mathfrak{F}$. A brief review of the case of an arbitrary normal lattice operator is given below, to be instantiated to the specific case of the implication, box and diamond operators in Sections~\ref{implication section}, \ref{necessity section} and \ref{possibility section}, respectively.

Given a frame $\mathfrak{F}=(s,Z,I,(R_j)_{j\in J},\sigma)$, each relation $R_j\subseteq Z_{j_{n(j)+1}}\times\prod_{k=1}^{n(j)}Z_{j_k}$ generates a sorted image operator, defined as in the Boolean case, except for the sorting
\begin{align}\label{sorted image ops}
F_j(\vec{W})&=\;\{w\in Z_{i_{n(j)+1}}\midsp \exists \vec{w}\;(wR_j\vec{w}\wedge\bigwedge_{s=1}^{n(j)}(w_s\in W_s))\} &=\; \bigcup_{\vec{w}\in\vec{W}}R_j\vec{w}.
\end{align}
Then $\mathbf{P}=(\largediamond:\powerset(Z_1)\leftrightarrows\powerset(Z_\partial):\lbbox, (F_j)_{j\in J} )$ is a 2-sorted powerset (poly)modal algebra, where $\largediamond$ and all $F_j$ are completely normal additive sorted operators and $\lbbox$ is completely multiplicative and normal. Alternatively, we represent $\mathbf{P}$ using the Gallois connection $\mathbf{P}=((\;)\rperp:\powerset(Z_1)\leftrightarrows\powerset(Z_\partial):{}\rperp(\;), (F_j)_{j\in J} )$, since the Galois connection and the residuated modal operators are interdefinable.

If $F_j$ is the (sorted) image operator generated by the frame relation $R_j$, let $\overline{F}_j$ be the closure of the restriction of $F_j$ to Galois sets (stable, or co-stable, according to sort).
\[
\xymatrix{
\prod_{k=1}^{n(j)}\powerset(Z_{j_k})\ar^{F_j}[rr]\ar^{(\;)''}@<0.5ex>@{->>}[d] && \powerset(Z_{j_{n(j)+1}})\ar^{(\;)''}@<0.5ex>@{->>}[d]\\
\prod_{k=1}^{n(j)}\mathcal{G}(Z_{j_k})\ar^{\overline{F}_j}[rr]\ar@<0.5ex>@{^{(}->}[u] && \mathcal{G}(Z_{j_{n(j)+1}})\ar@<0.5ex>@{^{(}->}[u]
}
\]
Then  $\overline{F}_j(\vec{G})=(F_j(\vec{G}))''$ and it follows that $\overline{F}_j$ is defined on a tuple $\vec{P}$ of Galois sets by $\overline{F}_j(\vec{P})=\bigvee_{\vec{w}\in\vec{G}}(R_j\vec{w})''$, where by $\vec{w}\in\vec{P}$ we mean the conjunction of coordinatewise membership statements $w_k\in P_k$, for $k=1,\ldots,n(j)$.

\begin{thm}
\label{dist from section stability}
The sorted operator $\overline{F}_j:\prod_{k=1}^{n(j)}\mathcal{G}(Z_{j_k})\lra \mathcal{G}(Z_{j_{n(j)+1}})$ distributes over arbitrary joins of Galois sets, in each argument place, returning a join in $\mathcal{G}(Z_{j_{n(j)+1}})$.
\end{thm}
\begin{proof}
The claim was proven in \cite[Theorem~3.12]{duality2}, using the smoothness property of the frame relation $R_j$.
\end{proof}
By the complete distribution property, $\overline{F}_j$ is residuated at each argument place and, from residuation, it follows that $\overline{F}_j$ is normal, i.e. $\overline{F}_j(\vec{G}[\emptyset'']_k)=\emptyset''$.

Note that for each $j\in J$ the sorted set operator $F_j:\prod_{k=1}^{n(j)}\powerset(Z_{j_k})\lra\powerset(Z_{j_{n(j)+1}})$ in the (sorted) powerset algebra $\mathbf{P}$ is completely additive (it distributes over arbitrary unions) in each argument place. Hence it is residuated, i.e. for each $1\leq k\leq n(j)$ there exists a set map $G_{j,k}$ such that $F_j(\vec{W}[V]_k)\subseteq U$ iff $V\subseteq G_{j,k}(\vec{W}[U]_k)$ which is defined by equation~\eqref{k-residual}
\begin{equation}\label{k-residual} 
G_{j,k}(\vec{W}[U]_k)=\bigcup\{V\subseteq Z_{j_k}\midsp F_j(\vec{W}[V]_k)\subseteq U\}.
\end{equation}

\begin{thm}\label{preservation of residuals}
If $G_{j,k}$ is the right residual of $F_j$ at the $k$-th argument place, then its restriction $\overline{G}_{j,k}$ to Galois sets is the right residual of $\overline{F}_j$ at the $k$-th argument place. Letting $P,Q,E$ range over Galois sets (and $\vec{P},\vec{Q}$ over tuples thereof) the right $k$-residual $\overline{G}_{j,k}$  of $\overline{F}_j$ can be defined in any of the equivalent ways in equation \eqref{residuals def}
  \begin{equation}\label{residuals def}
  \begin{array}{ccl}
  \overline{G}_{j,k}(\overline{P}[Q]_k) &=& \bigcup\{E\in\mathcal{G}(Z_{j_k})\midsp F_j(\vec{P}[E]_k)\subseteq Q\} \\
  &=&  \bigcup\{\Gamma u\in\mathcal{G}(Z_{j_k})\midsp F_j(\vec{P}[\Gamma u]_k)\subseteq Q\}  \\
  &=& \{u\in Z_{j_k}\midsp F_j(\vec{P}[\Gamma u]_k)\subseteq Q\}.
  \end{array}
  \end{equation}
\end{thm}
\begin{proof}
With the smoothness assumption on the relation $R_j$, for each $j\in J$, Theorem~\ref{dist from section stability} established that the (sorted) operator $\mathfrak{R}(F_j)=\overline{F}_j$ distributes over arbitrary joins in each argument place $k$. Hence it is residuated  and the residual is canonically defined by $\overline{G}_{j,k}(\overline{P}[Q]_k) = \bigvee\{E\in\mathcal{G}(Z_{j_k})\midsp F_j(\vec{P}[E]_k)\subseteq Q\}$, where the join is taken in $\mathcal{G}(Z_{j_k})$. The first line of equation~\eqref{residuals def} means that the join in question is actually a union. That the three lines of the equation are equivalent ways of defining the residual was proven in \cite[Proposition 3.14, Lemma 3.15]{duality2}.
\end{proof}

The Galois connection is a dual isomorphism of the complete lattices of stable and co-stable sets, $(\;)\rperp:\mathcal{G}(Z_1)\iso\mathcal{G}(Z_\partial)^\mathrm{op}:\rperp(\;)$. This allows for extracting single-sorted operators $\overline{F}_j^1:\prod_{k=1}^{n(j)}\mathcal{G}(Z_1)\lra\mathcal{G}(Z_1)$ and $\overline{F}_j^\partial:\prod_{k=1}^{n(j)}\mathcal{G}(Z_\partial)\lra\mathcal{G}(Z_\partial)$, by composition with the Galois connection maps
\begin{equation}\label{2single-sorted}
  \overline{F}^1_j(A_1,\ldots,A_{n(j)})=
  \left\{
  \begin{array}{cl}
  \overline{F}_j(\ldots,\underbrace{A_k}_{j_k=1},\ldots,\underbrace{A'_r}_{j_r=\partial},\ldots) & \mbox{if } j_{n(j)+1}=1
  \\
  (\overline{F}_j(\ldots,\underbrace{A_k}_{j_k=1},\ldots,\underbrace{A'_r}_{j_r=\partial},\ldots))' & \mbox{if } j_{n(j)+1}=\partial.
  \end{array}
  \right.
\end{equation}
From the residual $\overline{G}_{j,k}$  of $\overline{F}_j$ we can similarly extract a single-sorted operator $\overline{G}^1_{j,k}$.

It follows from Theorem~\ref{dist from section stability} and the definition of $\overline{F}^1_j$, that if the sort of $R_j$ is $\sigma(j)=(j_{n+1};j_1\cdots j_n)$, then $\overline{F}^1_j$ is a normal lattice operator of distribution type $\delta(j)=(j_1,\ldots,j_n;j_{n+1})$ (co)distributing over arbitrary meets/joins according to its distribution type.

\begin{defn}\label{full complex algebra defn}
For a frame $\mathfrak{F}=(s,Z,I,(R_j)_{j\in J},\sigma)$, its full complex algebra $\mathfrak{F}^+$ is defined as the normal lattice expansion $\mathfrak{F}^+=(\gpsi,\subseteq,\bigcap,\bigvee,\emptyset'', Z_1,(\overline{F}^1_j)_{j\in J})$. 
\end{defn}

\begin{rem}[Generalized J\'{o}nsson-Tarski Framework]\label{jt rem2}\mbox{}\\
If the frame $\mathfrak{F}=(s,Z,I,(R_j)_{j\in J},\sigma)$ is a classical Kripke frame (see Remark~\ref{jt rem1}), i.e. $Z_1=Z_\partial$, with $I\subseteq Z_1\times Z_1$ being the identity relation, then the Galois connection is set-complementation and the closure operator is the identity operator on subsets of $Z_1$. 
In that case, equation~\eqref{sorted image ops} simply defines the J\'{o}nsson-Tarski image operators \cite{jt1}. 

Operations on Boolean algebras  that are not additive arise by appropriately composing with set complementation. In generalizing the J\'{o}nsson-Tarski framework, we include all normal lattice operators, but notice that in the special case of Kripke frames the maps $\overline{F}_j$ we define as closure of the restriction to stable sets are identical to the $F_j$, since every subset is stable and closure means taking the complement of the complement. Since there is only one sort, in the sense that $Z_1=Z_\partial$, the maps $F_j$ are single-sorted normal additive operators (J\'{o}nsson-Tarski operators). In equation~\eqref{2single-sorted}, composition with the set-complementation operation is actually performed, in the case of a classical Kripke frame, returning a set operation that either distributes, or co-distributes over either intersections or unions (which are the same as joins of stable sets in the Kripke frame case). 

For an example, consider the case of defining the box operator, of distribution type $(\partial;\partial)$. The relation $R^{\partial\partial}_\Box$ generates an additive (diamond) operator $\ldminus$ on $\powerset(Z_\partial)=\powerset(Z_1)$,  $F_j(U)=\ldminus U$. To obtain the box operator on stable subsets of $Z_1$ (which in the Kripke frame case are all subsets of $Z_1$) we compose appropriately with the Galois connection, i.e. we set $\lbminus A=(\ldminus A')'$. But in the case of a classical Kripke frame the Galois connection is simply the set-complement operation, hence we  obtain the classical definition $\lbminus U=-(\ldminus(-U))$. 
\end{rem}

\subsection{Soundness of the Fusion Logic $\mathbf{\Lambda}_{\Box\Diamond}^\ra$ in the Relational Semantics}
\label{soundness for minimal}

\subsubsection{Implication}
\label{implication section}
\begin{prop}\label{soundness for implication}
The single-operator minimal logic $\mathbf{\Lambda}^\ra$ (the logic of implicative lattices) is sound in the relational semantics.
\end{prop}
\begin{proof}
This proof is a brief review based on \cite{choiceFreeHA}, where representation and duality results for implicative lattices and Heyting algebras were presented.

Let $\mathfrak{F}=(s,Z,I,R_\ra,\ldots,\sigma)$ be a frame where $\sigma(R_\ra)=(\partial:1\partial)$. To make the connection with \cite{choiceFreeHA} easier, where the case of implicative lattices was discussed, we henceforth use $T$ for the relation $R_\ra$, occasionally displaying its sort as a superscript $T^{\partial 1\partial}\subseteq Z_\partial\times(Z_1\times Z_\partial)$. The relation $T$ generates an image operator $F_\ra=\largetriangleright : \powerset(Z_1)\times\powerset(Z_\partial)\lra\powerset(Z_\partial)$ (therefore an additive and normal set operator), 
defined as in equation~\eqref{sorted image ops}, instantiating  to \eqref{arrow image op}
\begin{equation}\label{arrow image op}
U{\largetriangleright} V=\{y\in Y\midsp\exists x,v(x\in U\;\wedge\;v\in V\;\wedge\;yTxv)\}=\bigcup_{x\in U}^{v\in V}Txv.
\end{equation}
The closure of its restriction to Galois sets (according to its distribution type) gives the sorted Galois-set operator 
$\overline{F}_\ra={\Mtright}:\gpsi\times\gphi\lra\gphi$.

By Theorem~\ref{dist from section stability},  $\Mtright$ distributes over arbitrary joins of stable sets in the first argument place and of co-stable sets in the second argument place, returning a join of co-stable sets. Composing with the Galois connection, we define $\Ra:\gpsi\times\gpsi^\partial\lra\gpsi^\partial$ by setting 
\begin{eqnarray}
A\Ra C &=& (A\Mtright C')'=(A\largetriangleright C')'''=(A\largetriangleright C')'\label{Ra def}.
\end{eqnarray}
A sentence $\varphi\ra\psi$ is then interpreted as $\val{\varphi\ra\psi}=\val{\varphi}\Ra\val{\psi}=(\val{\varphi}\Mtright\yvval{\psi})'$.
In \cite[Proposition~3.6]{choiceFreeHA} the following were proven, where $T'$ is the Galois dual  of $T$.
\begin{equation}\label{Ra prop}
\begin{array}{ll}
  1. & \left(\bigvee_{i\in I}A_i\right)\Ra\left(\bigcap_{j\in J}C_j\right)=\bigcap_{i\in I,j\in J}(A_i\Ra C_j)\\
  2. & A\Ra C=\bigcap_{x\in A, C\upv y}(\Gamma x\Ra{}\rperp\{y\})\\
  3. & uT'xy$ iff $u\in(\Gamma x\Ra{}\rperp\{y\})$, for all $u,x\in X$ and $y\in Y\\
  4. & u\in (A\Ra C)$ iff $\forall x\in Z_1\forall y\in Z_\partial(x\in A\;\wedge\;C\upv y\lra uT'xy)\}.
\end{array}
\end{equation}
Item 4 determines the first semantic clause for implication that we have presented. By the result in item 1, the operation $\Ra$ has the same distribution type as $\ra$. Normality of $\Ra$ follows from normality of the additive operator $\Mtright$, given the definition of $\Ra$. Since the axiomatization of the minimal logic $\mathbf{\Lambda}^\ra$ only enforces normality and (co)distribution properties for $\ra$, we may conclude that $\mathbf{\Lambda}^\ra$ is sound in frames $\mathfrak{F}=(s,Z,I,R^{\partial\partial}_\Box,R^{11}_\Diamond,T^{\partial 1\partial})$.
\end{proof}
It remains to justify our claim that the two satisfaction clauses given for $\ra$ in Table~\ref{sat} are indeed equivalent. The proof reveals that the stable set operation $\Ra$ is the restriction to stable sets of a set operator $\Ra_T$, residuated with a product operator $\bigodot$ (i.e. the restriction of $\Ra_T$ to stable sets returns a stable set).

Let $\bigodot$ be the binary image operator generated by the derived relation $R^{111}$ (Definition~\ref{derived relations defn}), defined on $\powerset(Z_1)\times\powerset(Z_1)$ by
\begin{equation}\label{odot defn}
U\mbox{$\bigodot$}V=\{x\in Z_1\midsp\exists u,z(u\in U\wedge z\in V\wedge xR^{111}uz)\}=\bigcup_{u\in U}^{z\in V}R^{111}uz.
\end{equation}
Being completely additive, $\bigodot$ is residuated (in both argument places)
and we let $\Ra_T$ be the residual $U\Ra_T W=\bigcup\{X\midsp U\bigodot X\subseteq W\}$, satisfying the residuation condition $V\subseteq U\Ra_T W$ iff $U\bigodot V\subseteq W$. Since $x\in U\Ra_T W$ iff $U\bigodot\{x\}\subseteq W$, a straightforward calculation, left to the interested reader, gives 
\begin{equation}\label{Lra membership condition}
  U\Ra_T W=\{x\in Z_1\midsp\forall z,u\in Z_1(u\in U\wedge zR^{111}ux\lra z\in W)\}
\end{equation}

Let $\bigovert$ be the closure of the restriction of $\bigodot$ to stable sets. 

\begin{prop}
\label{Ra long and short}
  The operators $\bigovert,\Ra$ are residuated. In other words, for any Galois stable sets $A,F,C$ we have $A\bigovert F\subseteq C$ iff $F\subseteq A\Ra C$. Consequently, $\Ra$ is the restriction to Galois stable sets of the residual $\Ra_T$ of $\bigodot$.
\end{prop}
\begin{proof}
  The first claim of this Proposition is part of the second claim of \cite[Proposition~3.10]{choiceFreeHA} and we refer the reader to it for the proof. The second claim follows by Theorem~\ref{preservation of residuals}, which is a special instance of \cite[Theorem~3.14]{duality2}.
\end{proof}

Given equations \eqref{Lra membership condition}, \eqref{Ra prop} and Proposition~\ref{Ra long and short}, we have the following equivalent characterizations of the stable set $A\Ra C$
\begin{equation}\label{Ra equivalent characterizations}
  \begin{array}{ccl}
  A\Ra C &=& \{x\in Z_1\midsp\forall u\in Z_1\forall y\in Z_\partial(u\in A\wedge C\upv y\lra xT'uy)\}\\
  &=& \{x\in Z_1\midsp\forall u,z\in Z_1(u\in A\wedge zR^{111}ux\lra z\in C)\}.
  \end{array}
\end{equation}

\begin{rem}\label{residual of tright}
  Since the sorted operator $\largetriangleright:\powerset(Z_1)\times\powerset(Z_\partial)\lra\powerset(Z_\partial)$ is completely additive in the left argument place, it has a left adjoint $\Da:\powerset(Z_1)\times\powerset(Z_\partial)\lra\powerset(Z_\partial)$, i.e. 
\[
U\largetriangleright V\subseteq Y\mbox{ iff }V\subseteq U\Da Y,\hskip5mm\mbox{ for } U\subseteq Z_1 \mbox{ and } V,Y\subseteq Z_\partial.
\]
It is defined by $U\Da Y=\bigcup\{W\subseteq Z_\partial\midsp U\largetriangleright W\subseteq Y\}$. 

We have defined $\Mtright$ as the closure of the restriction of $\largetriangleright$ to Galois sets. In other words, $A\Mtright B=(A\largetriangleright B)''$. 

By Proposition~\ref{preservation of residuals}, the restriction of $\Da$ to Galois sets is the right residual of $\Mtright$ in its first argument place, i.e. $A\Mtright B\subseteq D$ iff $B\subseteq A\Da D$, and it is defined by 
$A\Da D=\{y\in Z_\partial\midsp A\Da\Gamma y\subseteq D\}$. 

Recall that in \eqref{Ra def} we defined $A\Ra C=(A\Mtright C')'=(A\largetriangleright C')'$. So we have
\begin{tabbing}
$F\subseteq A\Ra C$ \hskip4mm\=iff\hskip2mm\= $F\subseteq (A\Mtright C')'$\\
\>iff\> $A\Mtright C'\subseteq F'$\\
\>iff\> $C'\subseteq A\Da F'$\\
\>iff\> $(A\Da F')'\subseteq C$
\end{tabbing}
By uniqueness of adjoints, the stable set operator $\bigovert$ can be alternatively defined from the residual $\Da$ of $\Mtright$, by $A\bigovert F=(A\Da F')'$.
\end{rem}

\subsubsection{Necessity}
\label{necessity section}
The relation $R_\Box=R^{\partial\partial}_\Box$ generates an image operator $F_\Box=\ldminus :\powerset(Z_\partial)\lra\powerset(Z_\partial)$ (therefore additive and normal), 
defined as in equation~\eqref{sorted image ops}, instantiating  to 
\begin{equation}\label{box image op}
\ldminus V \;=\; \{v\in Z_\partial\midsp\exists y\in Z_\partial(vR_\Box y\mbox{ and }y\in V)\}.
\end{equation}

The closure of its restriction to Galois sets (according to its distribution type) gives the sorted Galois set operator 
\begin{eqnarray*} 
\overline{F}_\Box=\ldminus^{\!\!\prime\prime} &:& \mathcal{G}(Z_\partial)\lra\mathcal{G}(Z_\partial)\hskip1cm \ldminus^{\!\!\prime\prime} B=(\ldminus B)''.
\end{eqnarray*}
By Theorem~\ref{dist from section stability},  $\overline{F}_\Box$ distributes over arbitrary joins of co-stable sets (members of the complete lattice $\mathcal{G}(Z_\partial)$). Define $\lbb:\gpsi\lra\gpsi$ from $\ldminus^{\!\!\prime\prime}$  by setting 
\begin{eqnarray}
\lbb A &=&(\ldminus^{\!\!\prime\prime} A')'=(\ldminus A')'''=(\ldminus A')'\label{lbb def}.
\end{eqnarray}

By Theorem~\ref{dist from section stability},  $\lbb\left(\bigcap_{q\in Q}A_q\right)=\bigcap_{q\in Q}\lbb(A_q)$. 

A more direct proof of the distribution property can be however given for the particular case of $\lbb$, as we show in the sequel. The interest of the alternative proof is that it shows that $\lbb$ is the restriction to stable sets of a multiplicative operation $\lbminus_R$ on $\powerset(Z_1)$ and it provides a direct proof of soundness for $\mathbf{\Lambda}_\Box$. We explain below.

Let $R'_\Box=R^{1\partial}_\Box$ be the Galois dual of $R_\Box$ and $R''_\Box\subseteq Z_1\times Z_1$ its double dual defined by $xR''_\Box =(xR'_\Box)'$, for any $x\in Z_1$. 

The relation $R''_\Box$ generates a dual image operator $\lbminus_R:\powerset(Z_1)\lra\powerset(Z_1)$ 
\[
\lbminus_R U=\{x\midsp \forall u(xR''_\Box u\lra u\in U)\}=\{x\midsp xR''_\Box\subseteq U\}.
\] 

\begin{prop}\label{box as restriction}
$\lbb$ is the restriction of  $\lbminus_R$ to stable sets.
\end{prop}
\begin{proof}
By definition, $\lbb A=(\overline{F}_\Box A')'=(F_\Box A')'''=(F_\Box A')'$. Since $F_\Box$ is the image operator generated on $\powerset(Z_\partial)$ by the relation $R_\Box\subseteq Z_\partial\times Z_\partial$, we obtain from $\lbb A=(F_\Box A')'$ that
\begin{tabbing}
  $x\in\lbb A$\hskip3mm \= iff \hskip3mm\=  $\forall y(\exists v(yR_\Box v$ and $A\upv v)\lra x\upv y)$    \\
   \> iff \> $\forall v(A\upv v\lra \forall y(yR_\Box v \lra y\in\{x\}\rperp))$\\
   \> iff \> $\forall v(A\upv v\lra (R_\Box v\subseteq \{x\}\rperp))$ \\
   \> iff \>  $\forall v(A\upv v\lra(\Gamma x\subseteq R'_\Box v))$\\
   \>iff\> $\forall v(A\upv v\lra xR'_\Box v)$\\
   \>iff\> $A'\subseteq xR'_\Box$\\
   \>iff\> $xR''_\Box\subseteq A$
\end{tabbing}
and given the definition of $\lbminus_R$, the proof is complete. Note that  $R_\Box v\subseteq \{x\}\rperp$ is equivalent to $\Gamma x\subseteq R'_\Box v$ by the fact that $R_\Box v$ is a closed element.
\end{proof}

\begin{prop}\label{soundness for box}
The single-operator minimal logic $\mathbf{\Lambda}_\Box$ is sound in the relational semantics.
\end{prop}
\begin{proof}
Since $\lbminus_R$ distributes over arbitrary intersections of subsets of $Z_1$, as a classical dual image operator, so does $\lbb$, over any intersections of stable sets. Normality of $\lbb$ follows from normality of the classical dual image operator $\lbminus_R$. It follows from the satisfaction clause for $\Box$ in Table~\ref{sat} that $\val{\Box\varphi}=\lbb\val{\varphi}$, hence soundness follows.
\end{proof}

For later use, we list here a Lemma on the monotonicity properties of the relation $R''_\Box$.
\begin{lemma}
\label{monotonicity props of R double prime}
The relation $R''_\Box$ is decreasing in the first and increasing in the second argument place, i.e. ${\leq}\circ R''_\Box\circ{\leq}\subseteq R''_\Box$.
\end{lemma}
\begin{proof}
By definition, $zR''_\Box=(zR'_\Box)'$, hence it is a Galois set and therefore increasing. For the first argument place, assume $u\leq z$ and let $zR''_\Box p$. To show that $uR''_\Box p$, let $y\in Z_\partial$ be arbitrary and such that $uR'_\Box y$. It suffices to get $p\upv y$. The relation $R'_\Box$, defined by setting for any $v\in Z_\partial$ $R'_\Box v=(R^{\partial\partial}_\Box v)'$ is a Galois set, hence increasing. Then from the hypotheses that $u\leq z$ and $uR'_\Box y$ we obtain $zR'_\Box y$. By the hypothesis that $zR''_\Box p$, the conclusion $p\upv y$ follows.
\end{proof}

Since $\lbminus$ is the restriction of $\lbminus_R$ to stable sets, we hereafter  just write $\lbminus$ for either. For stable sets we typically use uppercase letters $A,C,F$ from the beginning of the alphabet (and $B,D,E$ for co-stable sets) and we use uppercase letters $U,V,W$ from the end of the alphabet for arbitrary subsets, therefore no confusion should hopefully arise. However, in some cases we may restore the  subscript $R$, for clarity purposes.

\subsubsection{Possibility}
\label{possibility section}
\begin{prop}\label{soundness for diamond}
The single-operator minimal logic $\mathbf{\Lambda}_\Diamond$ is sound in the relational semantics.
\end{prop}
\begin{proof}
The relation  $R_\Diamond=R^{11}_\Diamond$ generates an image operator $F_\Diamond=\ldvert :\powerset(Z_1)\lra\powerset(Z_1)$ (therefore additive and normal) defined as in equation~\eqref{sorted image ops}, instantiating  to 
\begin{equation}\label{diamond image op}
\ldvert U \;=\; \{z\in Z_1\midsp\exists x\in Z_1(zR_\Diamond x\mbox{ and }x\in U)\}.
\end{equation}

The closure of its restriction to Galois sets (according to its distribution type) gives the sorted Galois set operator 
$\overline{F}_\Diamond=\ldd :\mathcal{G}(Z_1)\lra\mathcal{G}(Z_1)$, defined by $\ldd A=(\ldvert A)''$.

By Theorem~\ref{dist from section stability}, $\overline{F}_\Diamond$ distributes over arbitrary joins of stable sets (members of the complete lattice $\mathcal{G}(Z_1)$), i.e.  $\ldd\left(\bigvee_{q\in Q}A_q\right)=\bigvee_{q\in Q}\ldd(A_q)$. Normality of $\ldd$ follows from normality of the classical image operator $\ldvert$.

The semantic clause for diamond is stated in terms of co-satisfaction (refutation), hence we compute the dual operator (in $\gphi$) which is defined by $\overline{F}^\partial_\Diamond(B)=(\overline{F}_\Diamond(B'))'$. Therefore, $\yvval{{\Diamond}\varphi}=\val{\Diamond\varphi}'=(\ldd\val{\varphi})'= (\ldvert\val{\varphi})'''=(\ldvert\val{\varphi})'$. By a computation similar to that in the proof of Proposition~\ref{box as restriction} we obtain $\val{\Diamond\varphi}'=\yvval{{\Diamond}\varphi}=\{y\midsp\forall v\in Z_\partial(yR''_\Diamond v\lra v\in\yvval{\varphi})\}$ and this establishes soundness.
\end{proof}

The double dual relation $R''_\Diamond\subseteq Z_\partial\times Z_\partial$, used in the above proof, is defined from $R^{11}_\Diamond$ in a similar way to $R''_\Box$. Namely, we let $R'_\Diamond\subseteq Z_\partial\times Z_1$ be the Galois dual of $R_\Diamond$ and we set $yR''_\Diamond=(yR'_\Diamond)'$.

For later use we define in $\gphi$ a box operator by $\lbvert B=(\ldvert B')'$. Letting also $\lbvert_R$ be the dual image operator generated in $\powerset(Z_\partial)$ be the relation $R''_\Diamond$, the proof of the following result is completely similar to the proof of Proposition~\ref{box as restriction} and it can be safely left to the reader. 
\begin{prop}
  \label{dual box as restriction}
  The co-stable set operator $\lbvert$ is the restriction of the dual image operator $\lbvert_R$ generated by the relation $R''_\Diamond$.\telos
\end{prop}

\section{Frames for Distribution-Free  $\mathbf{K}^\ra_{\Box\Diamond}$}
\label{minimal normal section}
For the fused modal system $\mathbf{\Lambda}_{\Box\Diamond}^\ra$ we have included the classical distribution and normality properties for implication and for both box and diamond in the axiomatization. It is, however, the system $\mathbf{K}_{\Box\Diamond}^\ra=\mathbf{\Lambda}_{\Box\Diamond}^\ra+\{\Box(p\ra q)\proves\Box p\ra\Box q\}$, which is the minimal distribution-free normal modal logic that raises some issues of interest.  

For distribution-free systems, the issue in question is, roughly, to find a frame class that will allow approximating the box operator both from above (upper box) and from below (lower box). In the canonical extension of a lattice, this issue corresponds to the question of whether the $\sigma$ (lower approximation) and $\pi$ (upper approximation) extensions of a lattice map can be shown to be identical.  

\subsection{Intermediate Structure and Lower and Upper MacNeille Extensions for $\lbb$}
\label{intermediate structures}
The set $Q=\{\Gamma x\midsp x\in Z_1\}\cup\{\rperp\{y\}\midsp y\in Z_\partial\}=Q_1\cup Q_\partial$ of principal (closed and open) elements of $\gpsi$ is partially ordered by inclusion. Let $\overline{Q}$ be its Dedekind-MacNeille completion. $\overline{Q}$ can be constructed as the family of sets that are stable under the Dedekind-MacNeille Galois connection $U=(U^\ell)^u$, where $U^u$, $U^\ell$ are the sets of upper bounds and lower bounds, respectively, of $U$. $\overline{Q}$ is uniquely characterized as the completion of $Q$ in which $Q$ is both join and meet dense  \cite[Theorem~7.41]{hilary_davey_2002}. Any Galois stable set can be equivalently approximated both from above and from below, since it is both the join of closed elements below it ($A=\bigvee_{x\in A}\Gamma x=\bigcup_{x\in A}\Gamma x$) and the meet of open elements in which it is contained ($A=\bigcap_{A\upv y}\rperp\{y\}=\bigcap\{\rperp\{y\}\midsp A\subseteq{}\rperp\{y\}\}$). It follows from this that $\overline{Q}\iso\gpsi$. The structure $Q=Q_1\cup Q_\partial$ will be referred to as the {\em intermediate structure}. 

What is now needed is to be able to approximate operations on stable sets both from above, taking their {\em upper MacNeille extension}, and from below, taking their {\em lower MacNeille extension},  having first defined them in the intermediate structure.

Recall that for any stable sets $A,C\in\gpsi$,   
\begin{tabbing}
$A\Ra C$ \hskip2mm\==\hskip1mm\= $\bigcap_{x\in A}^{C\upv y}(\Gamma x\Ra{}\rperp\{y\})$ \hskip2.5cm \= by \eqref{Ra prop} (\cite[Proposition~3.6]{choiceFreeHA})\\
$\lbb A$\>=\>$\lbb\bigcap_{A\upv y}{}\rperp\{y\}=\bigcap_{A\upv y}\lbb\rperp\{y\}$ \> by Proposition~\ref{box as restriction}\\
${\ldd}A$ \>=\> $\bigvee_{x\in A}{\ldd}\Gamma x$ \> by Theorem~\ref{dist from section stability}.
\end{tabbing}
Since $A\upv y$ iff $A\subseteq{}\rperp\{y\}$, the box operator is approximated from above (by computing first its values on open elements that cover its argument $A$). 
We write $\lbb^u$ for $\lbb$ when we need to distinguish it from the lower approximation $\lbb^\ell$ of the box operator, technically defined in Definition~\ref{lower box defn} so as to satisfy $\lbb^\ell\Gamma x=\lbb^u\Gamma x$, $\lbb^\ell\rperp\{y\}=\lbb^u\rperp\{y\}$ and $\lbb^\ell A=\bigvee_{\Gamma x\subseteq A}\lbb^\ell\Gamma x$. The extensions $\lbb^\ell, \lbb^u$ are the {\em lower and upper MacNeille extensions} \cite{mai-yde-harding} of their restriction to elements of the intermediate structure.

Membership in each of $\Gamma x{\Ra}\rperp\{y\}, \lbb\rperp\{y\}$ and $\ldd\Gamma x$ is definable by the respective conditions shown below,
\begin{tabbing}
$z\in(\Gamma x\Ra{}\rperp\{y\})$\hskip1cm\= iff\hskip3mm\= $zT'xy$\\
$z\in\lbb\rperp\{y\}$ \>iff\> $zR''_\Box\upv y$\\
$z\in \ldd\Gamma x$ \> iff\> $\forall y(yR''_\Diamond\subseteq\Gamma x\lra z\upv y)$.
\end{tabbing}

To ensure soundness of the K-axiom, it suffices to fulfill the prerequisite that the axiom holds in the intermediate structure and that the upper MacNeille extension of the box operator is contained in its lower extension, $\lbb^u A\subseteq\lbb^\ell A$, for any $A\in\gpsi$.

Note that, if the full complex algebra of the frame is a canonical extension of its subalgebra of clopen elements (assuming enough axioms are included so that we can prove that the clopens form a subalgebra of the full complex algebra of the frame), then it can be shown that identity $\lbb^u=\lbb^\ell$ actually holds. This can be carried out in detail, but as a second-order axiomatization of frames will be necessarily involved, we prefer to take a more direct (and simpler) approach.

\subsection{Refined Frames}
\label{refined frames section}
A frame in the frame class defined in Table~\ref{refined frames axioms} will be called a {\em refined frame}. For our current purposes in this article, $R_j$ is any one of the frame relations $R^{\partial\partial}_\Box$, $R^{11}_\Diamond$, or $T^{\partial 1\partial}$. The relation  $S^{11}_\Box$ mentioned in axiom (F5) is defined below.
\begin{defn}\label{S11box defn}
 Define the relation $S^{11}_\Box$ by $zS^{11}_\Box x$ iff $zR''_\Box\subseteq\Gamma x$. Equivalently, $zS^{11}_\Box x$ iff $z\in\lbb\Gamma x$. 
\end{defn}
\begin{defn}\label{lower box defn}
Define the lower approximation $\lbb^\ell$ of the box operator on stable sets by $\lbb^\ell\Gamma x=S^{11}_\Box x$ and extend to any stable set $A$ by setting $\lbb^\ell A=\bigvee_{x\in A}\lbb^\ell\Gamma x$.
\end{defn}

\begin{table}[!htbp]
\caption{Axiomatization of Refined Frames $\mathfrak{F}=(s,Z,I,(R_j)_{j\in J},\sigma)$}
\label{refined frames axioms}
($R_j\in\{R^{\partial\partial}_\Box,R^{11}_\Diamond,T^{\partial 1\partial}\}$)
\begin{tabbing}
  (F1) \hskip5mm\= The frame is separated. \\
  (F2)  \> Every frame relation $R_j$ is increasing in the first and decreasing in every other\\
  \> argument place. \\
  (F3)  \> For every tuple $\vec{p}$ of the proper sort,  $R_j\vec{p}$ is a closed element $\Gamma(\widehat{j}(\vec{p}))$\\
  (F4) \> Every frame relation $R_j$ is smooth.\\
  (F5) \> For each $x\in Z_1$, the section $S^{11}_\Box x$ of the relation defined by $zS^{11}_\Box x$ iff $zR''_\Box\subseteq\Gamma x$ is \\
  \>a closed element of $\gpsi$\\
  (F6) \> (Restricted join-distributivity of $\lbb$)
For all stable sets $A$, $\lbb\bigvee_{x\in A}\Gamma x\subseteq\bigvee_{x\in A}\lbb\Gamma x$.
\end{tabbing}
\hrule
\end{table}

Since $\bigvee_{x\in A}\Gamma x=A$ and $\lbb^u\Gamma x=\{u\midsp uR''_\Box\subseteq\Gamma x\}=\{u\midsp uS^{11}_\Box x\}=S^{11}_\Box x=\lbb^\ell\Gamma x$, axiom (F6) is equivalent to the inclusion $\lbb^uA\subseteq\lbb^\ell A$, for any stable set $A$.

\begin{rem}\label{second-order rem}
Axiom (F6) is second-order. We leave it as an open problem whether a first-order axiomatization of frames validating the K-axiom in the distribution-free case can be found. The situation is quite different for the case of modal extensions of the Intuitionistic propositional calculus, as shown in Section~\ref{distributive/Heyting section}.
\end{rem}

For each $j\in J$, we may define a map $\widehat{j}(\vec{u})=v$ iff $R_j\vec{u}=\Gamma v$, given (F3). By separation (F1) the map $\widehat{j}$  is well-defined, since $v$ is unique, for each tuple $\vec{u}$. In other words,
\begin{equation}\label{point ops and relations}
wR_j\vec{u}\mbox{ iff }\widehat{j}(\vec{u})\leq w, \mbox{ equivalently } R_j\vec{u}=\Gamma(\widehat{j}(\vec{u})).
\end{equation}

\begin{defn}[Point Operators]\label{point operators defn}
Let $\widehat{j}_\ra,\widehat{j}_\Box$ and $\widehat{j}_\Diamond$ be the point operators of equation~\eqref{point ops and relations} corresponding to the relations $T,R_\Box$ and $R_\Diamond$, respectively. To make arguments more readable, we rename them, using mnemonic notation, by setting $\widehat{j}_\ra=\tright,\widehat{j}_\Box=\boxminus$ and $\widehat{j}_\Diamond=\diamondvert$. Given axiom (F5), we let also $S^{11}_\Box x=\Gamma(\boxminus x)$, thus using the point operator $\boxminus$ as a sorted map on $Z_1\cup Z_\partial$.
\end{defn}
It is straightforward to verify that $\tright,\boxminus,\diamondvert$ are monotone in each of their argument places.

\subsection{K-Frames}
\label{K-frames section}
\begin{defn}\label{K-frames defn}
  A  refined frame  $\mathfrak{F}=(s,Z,I,R_\Box,R_\Diamond,R_\ra,\sigma)$ is a {\em K-frame} provided that the additional axiom (FK) below holds
  \begin{tabbing}
  (FK) \hskip5mm\= $\forall x\in Z_1\forall y,v,w\in Z_\partial[vT^{\partial 1\partial}(\boxminus  x)(\boxminus y)\;\wedge\; w\leq x{\triangleright}y)\lra vR^{\partial\partial}_\Box w]$.
  \end{tabbing}
\end{defn}
Recall that $vT^{\partial 1\partial}zy$ holds iff $(z{\triangleright}y)\leq v$.

\begin{prop}
\label{K-prop}
For any stable sets $A,C$ in a K-frame, $\lbb(A\Ra C)\subseteq(\lbb A\Ra\lbb C)$.
\end{prop}
\begin{proof}
If the frame is a K-frame, then (FK) directly implies the point inequality $\boxminus(x{\triangleright} y)\leq \boxminus  x{\triangleright}\boxminus y$. By  the properties in~\eqref{Ra prop}, $\Gamma x\Ra{}\rperp\{y\}=T'xy={}\rperp\{x{\triangleright} y\}$, where $T'$ is the Galois dual relation of the frame relation $T^{\partial 1\partial}$ and ${{\triangleright}}:Z_1\times Z_\partial\lra Z_\partial$ is the point operator defined by $x{\triangleright} y=v$ iff $Txy=\Gamma v$. Hence we obtain
\begin{eqnarray*}
\lbb(\Gamma x\Ra{}\rperp\{y\}) &=&\lbb(\rperp\{x{\triangleright} y\})={}\rperp\{\boxminus(x{\triangleright} y)\}\\
\lbb\Gamma x\Ra\lbb({}\rperp\{y\})&=& \Gamma(\boxminus  x)\Ra{}\rperp\{\boxminus y\}= {}\rperp\{\boxminus  x{\triangleright}\boxminus y\}
\end{eqnarray*}
By definition of the order, the inequality $\boxminus(x{\triangleright} y)\leq \boxminus  x{\triangleright}\boxminus y$ is equivalent to the inclusion ${}\rperp\{\boxminus(x{\triangleright} y)\}\subseteq {}\rperp\{\boxminus  x{\triangleright}\boxminus y\}$. It follows that the K-axiom holds in the intermediate structure, 
\begin{equation}\label{intermediate K}
\lbb(\Gamma x\Ra{}\rperp\{y\})\subseteq \lbb\Gamma x\Ra\lbb({}\rperp\{y\}).
\end{equation} 

From~\eqref{Ra prop}, given stable sets $A,C\in\gpsi$, we have  $\lbb(A{\Ra} C)=\bigcap_{x\in A}^{C\upv y}\lbb(\Gamma x{\Ra}{}\rperp\{y\})$, given that $A\Ra C=(\bigvee_{x\in A})\Ra(\bigcap_{C\upv y}{}\rperp\{y\})=\bigcap_{x\in A}^{C\upv y}(\Gamma x\Ra{}\rperp\{y\})$ and that $\lbb$ distributes over arbitrary intersections. 
The following calculation proves the claim,
\begin{tabbing}
$\lbb(A\Ra C)$\hskip3mm\==\hskip2mm\=$\bigcap_{x\in A}^{C\upv y}\lbb(\Gamma x\Ra{}\rperp\{y\})$\hskip2cm\= Properties in~\eqref{Ra prop}\\
\>$\subseteq$\> $\bigcap_{x\in A}^{C\upv y}(\lbb\Gamma x\Ra\lbb{}\rperp\{y\})$ \> By~\eqref{intermediate K}\\
\>=\> $(\bigvee_{x\in A}\lbb\Gamma x)\Ra(\bigcap_{C\upv y}\lbb{}\rperp\{y\})$\> Properties in~\eqref{Ra prop}\\
\>=\>  $(\bigvee_{x\in A}\lbb\Gamma x)\Ra\lbb(\bigcap_{C\upv y}{}\rperp\{y\})$\> $\lbb$ distributes over meets\\
\>$\subseteq$\> $(\lbb \bigvee_{x\in A}\Gamma x)\Ra\lbb C$\>  By axiom (F6)\\
\>=\> $\lbb A\Ra\lbb C$ \>
\end{tabbing}
and this completes the proof.
\end{proof}

\section{Frames for the Standard Extensions of $\mathbf{K}^\ra_{\Box\Diamond}$}
\label{extensions section}
\subsection{D-Frames}
\label{KD-frames section}
Let $\mathfrak{F}$ be a refined frame. If the additional axiom
\begin{tabbing}
(FD)\hskip5mm\= $\forall x\in Z_1\; S^{11}_\Box x\subseteq R_\Diamond x$
\end{tabbing}
holds in the frame, then we call $\mathfrak{F}$ a {\em D-frame}.

\begin{lemma}\label{D lemma}
The following are equivalent
\begin{enumerate}
  \item (FD)
  \item $\forall x\in Z_1\;{\diamondvert}x\leq\boxminus x$
  \item $\forall x\in Z_1\;(\boxminus x)R_\Diamond x$
  \item $\forall x\in Z_1\;\lbb\Gamma x\subseteq\ldvert\Gamma x$
\end{enumerate}
\end{lemma}
\begin{proof}
Equivalence of the first three is immediate, from definitions. Since the frame is refined, $\lbb\Gamma x\subseteq\lbb^\ell\Gamma x=\Gamma(\boxminus x)=S^{11}_\Box x$. Since also $\ldvert\Gamma x=R_\Diamond x$, equivalence of (3) and (4) is clear, too. 
\end{proof}

\begin{prop}  \label{D-prop}
In a D-frame, for any $A\in\gpsi$ it holds that $\lbb A\subseteq\ldd A$.
\end{prop}
\begin{proof}
In a D-frame $\lbb A\subseteq\bigvee_{x\in A}\lbb\Gamma x\subseteq\bigvee_{x\in A}\ldvert\Gamma x=\bigvee_{x\in A}\ldd\Gamma x\subseteq \ldd A$, using Lemma~\ref{D lemma}. 
\end{proof}

\subsection{Axioms T, S4 and Reflexive/Transitive Frames}
\label{reflexive section}
A {\em flat extension} of the system $\mathbf{\Lambda}_{\Box\Diamond}^\ra=\mathbf{\Lambda}_\Box\oplus\mathbf{\Lambda}_\Diamond\oplus\mathbf{\Lambda}^\ra$ is an extension by axioms introduced in a constituent logic of the fusion, hence involving a single operator. This is the case of the T and S4 axioms for each of the operators $\Box,\Diamond$, but also axioms strengthening the lattice base for implication, such as including a rule ${\small\infrule{\varphi\proves\psi}{\proves\varphi\ra\psi}}$, whose algebraic counterpart is an integrality constraint in the lattice: $a\leq b$ iff $1\leq a\ra b$. Or requiring that the lattice be distributive, or a Heyting algebra. For such extensions, involving the lattice base axioms only, we refer the reader to Section~\ref{distributive/Heyting section} and to \cite{choiceFreeHA}. 

Modeling the logics of implicative modal lattices that assume any of the T, B or S4 axioms is particularly simple and proofs can be carried out assuming only axioms (F1)--(F4) for refined frames in Table~\ref{refined frames axioms}. 

Let $\Gamma$ be a subset of the axioms in the following list.
\begin{tabbing}
(T$\Diamond$)\hskip4mm\= $\varphi\proves\Diamond\varphi$ \hskip1cm\=(T$\Box$)\hskip4mm \= $\Box\varphi\proves\varphi$\\
(S4$\Diamond$)\> ${\Diamond\Diamond}\varphi\proves\Diamond\varphi$ \>(S4$\Box$)\> $\Box\varphi\proves{\Box\Box}\varphi$
\end{tabbing}
We already know that each of the logics $\mathbf{\Lambda}_{\Box\Diamond}^\ra\oplus \Gamma$ is canonical, as this is a consequence of \cite[Theorem~6.3]{mai-harding}, in which Gehrke and Harding established that if a lattice equation involves only operators (distributing over joins), or only dual operators (distributing over meets), then the equation is preserved under canonical extensions. 

In this section we identify the corresponding classes of frames that validate $\Gamma$. 

A fast answer to the correspondence problem for axioms in $\Gamma$ can be given by considering the powerset algebras $(\powerset(Z_1),\lbb,\ldvert), (\powerset(Z_\partial),\lbvert,\ldminus)$, where recall that $\ldvert,\ldminus$ are the image operators generated by the relations $R^{11}_\Diamond, R^{\partial\partial}_\Box$, respectively, that we defined $\lbb A=(\ldminus A')'$, $\lbvert B=(\ldvert B')'$ and that we have shown that $\lbb,\lbvert$ are the restrictions to their respective lattices of Galois sets of the powerset dual operators $\lbminus_R,\lbvert_R$ generated by the double dual relations $R''_\Box, R''_\Diamond$, respectively (and we dropped the subscript $R$, as not necessary). Thus the T and S4 axioms are valid for $\lbb$ iff its generating relation $R''_\Box$ is reflexive, or transitive, respectively.  Similarly for $\lbvert$. Classically, it suffices to state only the T or S4 axiom for box and derive it for diamond by duality ($\Diamond=\neg\Box\neg$). Essentially the same proof can be given in the distribution-free setting. 

Indeed, assume for example that the relation $R''_\Diamond\subseteq Z_\partial\times Z_\partial$ is reflexive, so that the dual operator $\lbvert:\gphi\lra\gphi$ satisfies the T-axiom, $\lbvert B\subseteq B$. Let $A=B'$. Then we have $\lbvert B\subseteq B$ iff $(\ldvert B')'\subseteq B$ iff $B'\subseteq (\ldvert B')''$ iff $A\subseteq (\ldvert A)''$ and recall that we have defined $\ldd:\gpsi\lra\gpsi$ by $\ldd A=(\ldvert A)''$. 

The argument just given completely settles the problem, so we have proved the following.
\begin{prop}
  \label{reflexive/transitive prop}
In $\gpsi$, $\lbb,\ldd$ satisfy the T-axiom iff $R''_\Box,R''_\Diamond$ are reflexive and they satisfy the S4-axiom iff the relations are transitive.\telos
\end{prop}

The next four Propositions offer alternative characterizations of the frames. We prove the first two and leave the proofs of the other two to the interested reader.

\begin{prop}
\label{reflexive box}
  The following are equivalent:
  \begin{enumerate}
  \item $\lbb A\subseteq A$, for any $A\in\gpsi$
  \item $\lbb{}\rperp\{y\}\subseteq{}\rperp\{y\}$, for all $y\in Z_\partial$
  \item  $\boxminus y\leq y$, for all $y\in Z_\partial$
  \item The frame relation $R_\Box\subseteq Z_\partial\times Z_\partial$ is reflexive
  \item The double dual $R''_\Box\subseteq Z_1\times Z_1$ of $R_\Box$ is reflexive
  \end{enumerate}
\end{prop}
\begin{proof}
Assuming $\lbb A\subseteq A$, we have in particular $\lbb{}\rperp\{y\}\subseteq{}\rperp\{y\}$. This is equivalent to ${}\rperp\{\boxminus y\}\subseteq{}\rperp\{y\}$, i.e. $\boxminus y\leq y$. Since $R_\Box y=\Gamma(\boxminus y)$, we obtain that $yR_\Box y$ holds, for any $y\in Z_\partial$, i.e. $R_\Box$ is reflexive. To show that this implies reflexivity of $R''_\Box$, first we calculate its defining condition. 

We have that $yR_\Box v$ iff $\boxminus v\leq y$ and then $uR'_\Box v$ iff $u\upv\boxminus v$. Consequently, $xR''_\Box z$ iff $\forall v(xR'_\Box v\lra z\upv v)$ iff $\forall v(x\upv\boxminus v\lra z\upv v)$. 

By the above, $xR''_\Box x$ holds iff $\forall v(x\upv\boxminus v\lra x\upv v)$. Let $v\in Z_\partial$ be such that $x\upv\boxminus v$, i.e. $xR'_\Box v$, which is equivalent, by definition, to $\forall y(yR_\Box v\lra x\upv y)$. The case assumption is that $R_\Box$ is reflexive, hence $vR_\Box v$ holds and then $x\upv v$ follows, hence $R''_\Box$ is indeed reflexive.

Reflexivity of $R''_\Box$ immediately implies that $\lbb A\subseteq A$, for any $A\in\gpsi$, by the fact that $x\in\lbb A$ iff $\forall z(xR''_\Box z\lra z\in A)$. 
\end{proof}

\begin{prop}
\label{transitive box}
  The following are equivalent
  \begin{enumerate}
    \item $\lbb A\subseteq \lbb\lbb A$, for all $A\in\gpsi$
    \item $\lbb{}\rperp\{y\}\subseteq\lbb\lbb{}\rperp\{y\}$
    \item $\boxminus y\leq{\boxminus\boxminus} y$, for all $y\in Z_\partial$
    \item $R_\Box\subseteq Z_\partial\times Z_\partial$ is transitive
    \item The double dual $R''_\Box\subseteq Z_1\times Z_1$ of $R_\Box$ is transitive.
  \end{enumerate}
\end{prop}
\begin{proof}
Assuming $\lbb A\subseteq \lbb\lbb A$, for any $A\in\gpsi$, we have in particular $\lbb{}\rperp\{y\}\subseteq\lbb\lbb{}\rperp\{y\}$, for any $y\in Z_\partial$. This is equivalent to ${}\rperp\{\boxminus y\}\subseteq{}\rperp\{{\boxminus\boxminus} y\}$ and then by definition of the order it follows that $\boxminus y\leq{\boxminus\boxminus} y$.

For transitivity of $R_\Box$, assume $yR_\Box v R_\Box w$, i.e. $\boxminus w\leq v$ and $\boxminus v\leq y$. Then by the case assumption and monotonicity of the point operator $\boxminus$, we obtain $\boxminus w\leq {\boxminus\boxminus} w\leq\boxminus v\leq y$, i.e. $yR_\Box w$ holds.

Note that, transitivity of $R_\Box$ implies that the inclusion $\lbb A\subseteq \lbb\lbb A$ holds, so that (1)--(3) are pairwise equivalent.

To show that $R''_\Box$ is transitive, assume $xR''_\Box uR''_\Box z$ and recall that $xR''_\Box z$ iff $\forall v(xR'_\Box v\lra z\upv v)$ iff $\forall v(x\upv \boxminus v\lra z\upv v)$. Let then $v\in Z_\partial$ be such that $x\upv\boxminus v$. By case assumption, $\boxminus v\leq{\boxminus\boxminus} v$ and since $\upv$ is increasing in both argument places we obtain that $x\upv{\boxminus\boxminus} v$. This means that $xR'_\Box\boxminus v$ and given the definition of $xR''_\Box$ as the Galois dual of $xR'_\Box$ and the assumption that $xR''_\Box u$,  we get $u\upv\boxminus v$. But $uR''_\Box z$, by assumption, and then given the defining condition above for $R''_\Box$ and our last conclusion that $u\upv\boxminus v$, it follows that $z\upv v$. Thereby $xR''_\Box z$ does indeed hold.

Finally, we assume that $R''_\Box$ is transitive and we show that if $A\in\gpsi$, then $\lbb A\subseteq \lbb\lbb A$ holds, where recalle that $x\in\lbb A$ iff $\forall z(xR''_\Box z\lra z\in A)$.

Let $x\in\lbb A$. If $x\not\in\lbb\lbb A$, let $u$ be such that $xR''_\Box u$, but $u\not\in\lbb A$. There is then $z$ such that $uR''_\Box z$ but $z\not\in A$. This contradicts the assumption $x\in\lbb A$ since by transitivity we obtain given $u$ and $z$ above that $xR''_\Box z$, hence $z\in A$.
\end{proof}

\begin{prop}
\label{reflexive diamond}
  The following are equivalent
\begin{enumerate}
  \item $A\subseteq\ldd A$, for all $A\in\gpsi$
  \item $\Gamma x\subseteq\ldd\Gamma x$, for all $x\in Z_1$
  \item $\diamondvert x\leq x$, for all $x\in Z_1$   
  \item $R_\Diamond$ is reflexive
  \item The double dual $R''_\Diamond\subseteq Z_\partial\times Z_\partial$ of $R_\Diamond$ is reflexive.
  \end{enumerate}
\end{prop}
\begin{proof}
  Left to the interested reader.
\end{proof}

\begin{prop}
\label{transitive diamond}
The following are equivalent
\begin{enumerate}
  \item $\ldd\ldd A\subseteq\ldd A$, for any stable set $A$
  \item $\ldd\ldd\Gamma x\subseteq\ldd\Gamma x$, for any point $x\in Z_1$
  \item $\diamondvert x\leq{\diamondvert\diamondvert} x$,  for any point $x\in Z_1$
  \item $R_\Diamond\subseteq Z_1\times Z_1$ is transitive
  \item The double dual $R''_\Diamond\subseteq Z_\partial\times Z_\partial$ of $R_\Diamond$ is transitive.
\end{enumerate}
\end{prop}
\begin{proof}
  Left to the interested reader.
\end{proof}

\begin{rem}\label{conradie-palmiziano rem}
In  Conradie and Palmigiano \cite{conradie-palmigiano} a correspondent for the T-axiom in a non-distributive setting is calculated, as an example of applying the ALBA algorithm. The authors work within the framework of RS-frames and calculate that the axiom $\Box\varphi\proves\varphi$ is valid in any RS-frame where the accessibility relation $S_\Box$ is a subrelation of the Galois relation $\upv$ (which they denote by $I$). This is only in appearance different from the reflexivity condition for $R_\Box$ that we have calculated. As explained in \cite{kata2z}, every example of relational modeling of a logical operator in a non-distributive setting that has been given in the RS-frames approach uses an accessibility relation $S$ that is the Galois dual relation $R'$ of the relation used in this author's approach. We leave it to the reader to verify that $R_\Box,R_\Diamond$ are reflexive iff, respectively, $R'_\Box\subseteq{\upv}$ and $(R'_\Diamond)^{-1}\subseteq{\upv}$. 
\end{rem}
\begin{rem}
A different approach is taken in Brezhanishvilli, Dmitrieva, de Groot and Morachini \cite{choice-free-dmitrieva-bezanishvili}, where weak, distribution-free, modal logic systems are discussed. The system of their focus is a distribution-free version of Dunn's \texttt{PML} (Positive Modal Logic) \cite{pml}. A single relation $R$ in frames generates both $\Box,\Diamond$, the interpretation for modal formulae is standard, except that in pursuing a way to define frames that validate at least one of Dunn's interaction axioms they find it necessary to abandon distribution of diamonds over joins. The authors do not work with canonical extensions, they model their logic(s) of interest in (topological) semilattices, they develop a related correspondence theory and they calculate that $\Box\varphi\proves\varphi$ and $\varphi\proves\Diamond\varphi$ are valid in semilattice frames satisfying the conditions $\forall x\exists y(xRy$ and $y\leq x)$ and $\forall x\exists y(xRy$ and $x\leq y)$, respectively. In other words, reflexivity is a property of the composite relations $R\circ{\leq}$ and $R\circ{\geq}$, respectively.
\end{rem}

\subsection{Axiom B and Residuated Box/Diamond }
\label{B section}
Let $\mathfrak{F}=(s,Z,I,R^{\partial\partial}_\Box,R^{11}_\Diamond, T^{\partial 1\partial})$, where we displayed the sort of the frame relations, be a frame satisfying at least axioms (F1)--(F4) of Table~\ref{refined frames axioms} and assume the frame also satisfies the following axiom
\begin{tabbing}
\hskip5mm\=(FB)\hskip2mm $\forall x\in Z_1\forall y\in Z_\partial(xR'_\Box y \leftrightarrow yR'_\Diamond x)$,
\end{tabbing}
in other words $R'_\Diamond=(R'_\Box)^{-1}$. Given definitions, axiom (FB) is equivalent to the  condition $\Gamma x\subseteq{}\rperp\{\boxminus y\}$ iff $\Gamma(\diamondvert x)\subseteq{}\rperp\{y\}$ in the intermediate structure. In turn, this is equivalent to the following condition, which is more convenient to use
\begin{tabbing}
\hskip5mm\=(FB)\hskip2mm $\forall x\in Z_1\forall y\in Z_\partial(x\upv\boxminus y \leftrightarrow \diamondvert x\upv y)$.
\end{tabbing}

\begin{prop}
\label{B prop in frames}
In a frame satisfying at least axioms (F1)--(F4) of Table~\ref{refined frames axioms} together with (FB) the operators $\lbb,\ldd$ in the full complex algebra of the frame are residuated, i.e. for all stable sets $A,C$ we have $A\subseteq\lbb C$ iff $\ldd A\subseteq C$.
\end{prop}
\begin{proof}
The claim is equivalent to $A\subseteq\bigcap_{C\upv y}{}\rperp\{\boxminus y\}$ iff $\bigvee_{x\in A}\Gamma(\diamondvert x)\subseteq C$ iff $C'\subseteq \bigcap_{x\in A}\{\diamondvert x\}\rperp$.

Assuming $A\subseteq\bigcap_{C\upv y}{}\rperp\{\boxminus y\}$, let $C\upv y$ and $x\in A$. From the assumption we get $x\upv\boxminus y$ and, by axiom (FB), this is equivalent to $\diamondvert x\upv y$. This shows $C'\subseteq \bigcap_{x\in A}\{\diamondvert x\}\rperp$. Conversely, from the hypothesis $C'\subseteq \bigcap_{x\in A}\{\diamondvert x\}\rperp$, if $x\in A$ and $C\upv y$ we obtain $\diamondvert x\upv y$ and then by axiom (FB) $x\upv\boxminus y$ follows. Thus $A\subseteq\bigcap_{C\upv y}{}\rperp\{\boxminus y\}$.
\end{proof}

\subsection{S5-Frames}
\label{K5-frames section}
Let $\mathfrak{F}$ be a refined frame (axiomatized in Table~\ref{refined frames axioms}) and assume the following frame axiom also holds
\begin{tabbing}
(FS5)\hskip5mm\= $\forall x\in Z_1\;[(\diamondvert x)R''_\Box\subseteq R^{11}_\Diamond x$ and $R^{11}_\Diamond(\boxminus x)\subseteq S^{11}_\Box x]$.
\end{tabbing}
Given the definition of $S^{11}_\Box$ and the fact that $R^{11}_\Diamond x=\Gamma(\diamondvert x)$, $(\diamondvert x)R''_\Box\subseteq R^{11}_\Diamond x$ can be equivalently restated as  $\forall x\in Z_1\;(\diamondvert x,\diamondvert x)\in S^{11}_\Box$.

\begin{lemma}
\label{5lemma}
For a stable set $A$, if $x\in A$, then ${\boxminus\diamondvert}x\in\lbb\ldd A$.
\end{lemma}
\begin{proof}
The hypothesis is equivalent to $\Gamma x\subseteq A$, from which by monotonicity we obtain that $\lbb\ldd\Gamma x\subseteq\lbb\ldd A$. The left-hand side is equal to $\Gamma({\boxminus\diamondvert}x)$,  hence if $x\in A$, indeed ${\boxminus\diamondvert}x\in\lbb\ldd A$.
\end{proof}

\begin{prop}
\label{5Prop}
Let $\mathfrak{F}$ be a refined frame satisfying, in addition, axiom (FS5). Then for any stable set $A$, the inclusion $\ldd A\subseteq\lbb\ldd A$ is true.
\end{prop}
\begin{proof}
Since $\ldd A=\bigvee_{x\in A}\ldd\Gamma x=\bigvee_{x\in A}\ldvert\Gamma x=\bigvee_{x\in A}\Gamma(\diamondvert x)$ it suffices to argue that for all $x\in A$, $\diamondvert x\in\lbb\ldd A$. By Lemma~\ref{5lemma}, ${\boxminus\diamondvert}x\in\lbb\ldd A$. The new axiom (FS5) is equivalent to the point inequality ${\boxminus\diamondvert}x\leq\diamondvert x$. Galois sets are increasing, hence $\diamondvert x\in\lbb\ldd A$.
\end{proof}

The dual case for the validity of $\ldd\lbb A\subseteq\lbb A$ is similar, left to the interested reader.

\section{Canonical Frame, Canonicity of Axioms and Completeness}
\label{canonical frame section}

\subsection{Canonical Frame Construction}
\label{canonical frame construction section}
\begin{thm}[Realizing Distribution Types]
\label{rep prop}
Let $\mathbf{L}_\tau=(L,\leq,\wedge,\vee,0,1,(f_j)_{j\in J})$ be a normal lattice expansion, where for each $j\in J$, $\delta(j)=(j_1,\ldots,j_{n(j)};j_{n(j)+1})$ is the distribution type of the normal lattice operator $f_j$.

Define the dual frame $(\mathbf{L}_\tau)_+=(s,Z,I,(R_j)_{j\in J},\sigma)$ of $\mathbf{L}_\tau$ as follows:
\begin{itemize}
  \item $s=\{1,\partial\}$, $Z=(Z_1,Z_\partial)$, where $Z_1=\filt(\mathbf{L}_\tau)$, $Z_\partial=\idl(\mathbf{L}_\tau)$ are the sets of non-empty lattice filters and ideals, respectively
  \item $xIy$ iff $x\cap y=\emptyset$ 
  \item For each $j\in J$ and where $\delta(j)=(j_1,\ldots,j_{n(j)};j_{n(j)+1})$ is the distribution type of $f_j$, the relation $R_j$ of sort $\sigma(j)=(j_{n(j)+1};j_1\cdots j_{n(j)})$ is defined as in the classical case, except for the sorting, by
      \[
      wR_j\vec{u} \mbox{ iff }\forall a_1,\ldots, a_n\in L(\bigwedge_{k=1}^{n(j)}(a_k\in u_k)\lra f_j\vec{a}\in w)
      \]
      and it is smooth.
\end{itemize} 
Then the full complex algebra $((\mathbf{L}_\tau)_+)^+$ of the dual frame $(\mathbf{L}_\tau)_+$ of the normal lattice expansion $\mathbf{L}_\tau$ is a canonical extension of $\mathbf{L}_\tau$ with embedding map sending a lattice element to the set of filters containing it and such that if $\delta(j)=(j_1,\ldots,j_{n(j)};1)$ (the operator $f_j$ outputs joins) then $\overline{F}_j^1$ is the $\sigma$-extension of $f_j$, else it is its $\pi$-extension (in the sense of \cite[Definition~4.1]{mai-harding}).
\end{thm}
\begin{proof}
A detailed exposition of bounded lattice representation and  duality is given in  \cite{choiceFreeStLog} and, in a context of implicative lattices, in \cite{choiceFreeHA}, both presenting special cases of the duality for normal lattice expansions in \cite{duality2}. The lattice representation and duality underlying the results of each of \cite{duality2,choiceFreeHA,choiceFreeStLog} first appeared in print in \cite{iulg,sdl}. The way we define point operators and canonical relations was first introduced in \cite{dloa}. Point operators, as we define them, are the same as the duals of lattice quasi-operators defined in the dual lattice space by Moshier and Jipsen \cite{Moshier2014a,Moshier2014b}, as we detailed in \cite[Remark~4.2, Remark~4.8]{duality2}. That the complete lattice of stable sets is identical to the complete lattice of the Moshier-Jipsen saturated sets (defined as the intersections of open filters that cover them) was proven in \cite[Proposition~4.6]{choiceFreeStLog}.

For an implicative modal lattice $\mathbf{L}$, following \cite{choiceFreeStLog,choiceFreeHA}, we let  $Z_1=\filt(\mathbf{L})$, $Z_\partial=\idl(\mathbf{L})$ be the sets of non-empty lattice filters and ideals, respectively, and we define the relation $I\subseteq\filt(\mathbf{L})\times\idl(\mathbf{L})$ by $xIy$ iff $x\cap y=\emptyset$ (with the Galois relation of the frame defined then by $x\upv y$ iff $x\cap y\neq\emptyset$). The underlying lattice of $\mathbf{L}$ is embedded in the complete lattice $\gpsi$ of stable sets of filters via the map $\zeta_1$ defined by $\zeta_1(a)=\{x\in\filt(\mathbf{L})\midsp a\in x\}=X_a$ and dually embedded in $\gphi$ via the map $\zeta_\partial(a)=\{y\in\idl(\mathbf{L})\midsp a\in y\}=Y^a$. Meets are taken to intersections and lattice joins to joins  (closures of unions) of stable sets in $\gpsi$. 

Note that $X_a=\{x\in\filt(\mathbf{L})\midsp x_a\subseteq x\}=\Gamma x_a$ and $Y^a=\{y\in\idl(\mathbf{L})\midsp y_a\subseteq y\}=\Gamma y_a$, where for a lattice element $e$ we write $x_e, y_e$ for the principal filter and ideal, respectively, generated by $e$. The sets $X_a$ are clopen elements of $\gpsi$, since $\Gamma x_a={}\rperp\{y_a\}={}\rperp(\Gamma y_a)$ (and then so are the sets $Y^a$ in $\gphi$). 

Taking the families $\mathcal{B}=\{X_a\midsp a\in\mathbf{L}\}$ and $\mathcal{C}=\{Y^a\midsp a\in\mathbf{L}\}$ as bases for topologies on $Z_1=\filt(\mathbf{L})$ and $Z_\partial=\idl(\mathbf{L})$, respectively, it was shown in \cite{choiceFreeStLog} that $\mathcal{B,C}$ generate spectral topologies (\cite[Proposition~4.2]{choiceFreeStLog}), that $\mathcal{B,C}$ are the families of compact open, Galois stable and co-stable sets, respectively, and that they are dually isomorphic lattices and sublattices of $\gpsi,\gphi$, respectively (\cite[Proposition~4.4]{choiceFreeStLog}). Furthermore, the map $a\mapsto X_a$ is a bounded lattice isomorphism (\cite[Theorem~4.5]{choiceFreeStLog}). 

That the full complex algebra of the dual frame of the lattice is a canonical extension of the lattice was proven in \cite[Proposition~2.6]{mai-harding}. 

The canonical relations $R_\Box^{\partial\partial},R^{11}_\Diamond$ and $R^{\partial 1\partial}_\ra=T$ are defined using the canonical point operators $\boxminus: Z_\partial\lra Z_\partial$, $\diamondvert:Z_1\lra Z_1$ and $\triangleright:Z_1\times Z_\partial\lra Z_\partial$, setting $yR_\Box v$ iff $\boxminus v\subseteq y$, $xR_\Diamond z$ iff $\diamondvert z\subseteq x$ and $yT^{\partial 1\partial}xv$ iff $x{\triangleright}v\subseteq y$, given the point operator definitions in~\eqref{canonical point operators}
\begin{equation}\label{canonical point operators}
\boxminus y=\bigvee\{y_{\Box a}\midsp a\in y\}\hskip1cm \diamondvert x=\bigvee\{x_{\Diamond a}\midsp a\in x\}\hskip1cm x{\triangleright}y=\bigvee\{y_{a\ra b}\midsp a\in x, b\in y\}.
\end{equation}
A point operator $\boxminus$, generating the relation $S^{11}_\Box$, is defined on filters by $\boxminus x=\bigvee\{x_{\Box a}\midsp a\in x\}$.

The reader can easily verify that the relations $R_\Box^{\partial\partial},R^{11}_\Diamond$ and $R^{\partial 1\partial}_\ra=T$ are equivalently defined by the respective instance of the clause in the statement of this Theorem. The general proof (for any normal lattice expansion and the associated canonical relations) of this fact is given in~\cite[Lemma~4.4]{duality2}. 

Operators $\lbb,\ldd$ and $\Ra$ are defined in the canonical frame from the frame relations as detailed in Section~\ref{frames introduction section} for any frame. That the so defined maps are the $\sigma$ (for $\ldd$) and $\pi$-extension (for $\lbb, \Ra$), depending on the output type of their distribution type was argued for, for any normal lattice expansion, in \cite[Proposition~28]{kata2z}.

That the lattice representation map $a\mapsto X_a$ is an isomorphism of implicative modal lattices follows from the observations below
\begin{tabbing}
$\Diamond a$\hskip7mm\=$\mapsto$\hskip2mm\= $X_{\Diamond a}=\Gamma x_{\Diamond a}=\Gamma(\diamondvert x_a)=\ldvert\Gamma x_a=\ldd\Gamma x_a=\ldd X_a$\\
$\Box a$\>$\mapsto$\> $X_{\Box a}=\Gamma x_{\Box a}={}\rperp\{y_{\Box a}\}={}\rperp\{\boxminus y_a\}=\lbb({}\rperp\{y_a\})=\lbb\Gamma x_a=\lbb X_a$\\
$a\ra b$\>$\mapsto$\> $X_{a\ra b}=X_a\Ra X_b$, by \cite[Proposition~4.8]{choiceFreeHA}.
\end{tabbing}

The axioms (F1)--(F4) for a refined frame (though the term was not used) were argued for in \cite{duality2}, for any normal lattice expansion. 

For axiom (F5), the defining condition for the relation $S^{11}_\Box$ in the canonical frame is $zS^{11}_\Box x$ iff $\boxminus x\subseteq z$, where we define $\boxminus  x$ to be the filter generated by the set $\{\Box a\midsp a\in x\}$, equivalently $\boxminus  x=\bigvee_{a\in x}x_{\Box a}$. To verify that $zS^{11}_\Box x$ iff $zR''_\Box\subseteq\Gamma x$ we verify that $S^{11}_\Box x=\lbb\Gamma x$. Since $\lbb\Gamma x=\bigcap_{x\upv y}\rperp\{y\}=\bigcap_{x\upv y}\rperp\{\boxminus y\}$, the claim to prove is equivalent to showing that for any filters $x,z$, $\boxminus  x\subseteq z$ iff $\forall y(x\upv y\lra z\upv\boxminus y)$, where in the canonical frame $\upv$ means non-empty intersection. If $\boxminus \subseteq z$ and $y$ is an ideal such that $a\in x\cap y\neq\emptyset$, then by $a\in x$ and $a\in y$ we obtain $\Box a\in\boxminus  x$ and $\Box a\in\boxminus y$. Hence $\Box a\in z\cap\boxminus y\neq\emptyset$. 

Conversely, assume $\forall y(x\upv y\lra z\upv\boxminus y)$, but suppose for a contradiction that $\boxminus  x\not\subseteq z$. Hence there is an element $a\in x$ such that $\Box a\not\in z$. But since $a\in x\cap y_a$ (where $y_a$ is the principal ideal generated by the element $a$), the case hypothesis implies that $z\upv \boxminus y_a$. Since $\boxminus y_a=y_{\Box a}$ we get $\Box a\in z$, contradiction.

Finally, axiom (F6) posits that $\lbb^u A\subseteq\lbb^\ell A$, for any stable set $A$. In the canonical frame $\lbb^u=\Box^\pi$ and $\lbb^\ell=\Box^\sigma$ are the $\pi$ and $\sigma$ extensions of the box operator. 

Having argued that the full complex algebra of the canonical frame is a canonical extension of the represented implicative modal lattice, identity of the $\sigma$ and $\pi$ extensions of the box operator holds, by \cite[Lemma~4.4]{mai-harding}. 

By the above, the proof that the canonical frame is a refined frame (axiomatized in Table~\ref{refined frames axioms}) is complete.
\end{proof}

It remains to examine the cases where additional axioms are satisfied in the implicative modal lattice.

For the lattice distribution axiom and for the case where the lattice is a Heyting algebra and the respective proofs that the full complex algebra of the canonical frame is (completely) distributive and a (complete) Heyting algebra, respectively, we refer the reader to Section~\ref{distributive/Heyting section} and to \cite[Proposition~4.7, Proposition~4.11]{choiceFreeHA}.

\subsection{Canonicity of the K, D, T, B and S4, S5 Axioms}
\begin{prop}[Canonicity of the K-Axiom]
The frame axiom (FK) holds in the canonical frame and it is equivalent to the inclusion
  $\lbb(\Gamma x\Ra{}\rperp\{y\})\subseteq\lbb\Gamma x\Ra\lbb({}\rperp\{y\})$,  for any filter $x$ and ideal $y$.
\end{prop}
\begin{proof}
It follows from the definition of the point operators that $c\in x{\triangleright}y$ iff there exist elements $a\in x, b\in y$ such that $c\leq a\ra b$ and then $e\in\boxminus(x{\triangleright} y)$ iff there are elements $a\in x, b\in y$ such that $e\leq\Box(a\ra b)$. Since the K-axiom is assumed for the lattice (logic), we obtain that $e\leq\Box(a\ra b)\leq\Box a\ra\Box b$. Since $\Box a\in\boxminus x$ and $\Box b\in\boxminus y$ we get $e\in\boxminus x{\triangleright}{\boxminus} y$. By the definition of the canonical relation $R^{\partial\partial}_\Box$ we may conclude that (FK) holds in the canonical frame.

The inclusion in the statement of the current Lemma is equivalent to the inclusion $\boxminus(x{\triangleright} y)\subseteq \boxminus x{\triangleright}{\boxminus} y$. 
\end{proof}

\begin{prop}[Canonicity of the D-Axiom]
  The canonical frame for the minimal logic $\mathbf{\Lambda}^\ra_{\Box\Diamond}+\{\Box\varphi\proves\Diamond\varphi\}$ is a D-frame. In other words, the D-axiom is canonical.
\end{prop}
\begin{proof}
Proving canonicity of the D-axiom is equivalent to proving that the inequality $\diamondvert x\subseteq\boxminus x$ holds for any filter $x$ in the canonical frame. Since $\diamondvert x=\bigvee_{a\in x} x_{\Diamond a}$, it is enough to show that for any filter $x$, if $a\in x$, then ${\Diamond} a\in\boxminus x$. But the latter is defined by $\boxminus x=\bigvee_{a\in x}x_{\Box a}$ and since the D-axiom holds in the lattice (logic), it is clear that if $a\in x$ then since $\Box a\in\boxminus x$ and $\Box a\leq \Diamond a$ we have $\Diamond a\in\boxminus x$. Hence (FD) holds in the canonical frame.
\end{proof}

To prove canonicity of the T and S4 axioms, it suffices by Propositions~\ref{reflexive box}, \ref{transitive box}, \ref{reflexive diamond} and \ref{transitive diamond} to prove the corresponding point inequations in the intermediate structure, as shown in the following list
\begin{tabbing}
(T$\Diamond$)\hskip4mm\= $\varphi\proves\Diamond\varphi$\hskip1.5cm\= $\diamondvert x\subseteq x$,\hskip1cm\= for all filters $x$\\
(S4$\Diamond$)\> ${\Diamond\Diamond}\varphi\proves\Diamond\varphi$\> $\diamondvert x\subseteq{\diamondvert\diamondvert}x$,\> for all filters $x$\\
(T$\Box$) \> $\Box\varphi\proves\varphi$\> $\boxminus y\subseteq y$,\> for all ideals $y$\\
(S4$\Box$)\> $\Box\varphi\proves{\Box\Box}\varphi$\> $\boxminus y\subseteq{\boxminus\boxminus}y$,\> for all ideals $y$.
\end{tabbing}

\begin{prop}[Canonicity of the T and S4 Axioms]\label{T4 prop canonical}
Let $\Delta$ be a subset of the above set of T and S4 axioms and $\mathbf{\Delta}=
\mathbf{\Lambda}^\ra_{\Box\Diamond}+\Delta$ the corresponding axiomatic extension of the minimal implicative modal logic. Then $\mathbf{\Delta}$ is canonical.
\end{prop}
\begin{proof}
By Propositions~\ref{definability}, \ref{completeness} and \ref{duality}, we may equivalently work with implicative modal algebras satisfying the corresponding axioms.

Assume $a\leq\Diamond a$ holds and let $x$ be any (nonempty) filter. By definition of $\diamondvert x=\bigvee_{a\in x}x_{\Diamond a}$ (where $x_b$ designates the principal filter generated by a lattice element $b$), $e\in\diamondvert x$ iff there is some lattice element $a\in x$ such that $\Diamond a\leq e$. But then $a\leq\Diamond a\leq e$, so that $e\in x$. This proves the inclusion $\diamondvert x\subseteq x$.

Assume ${\Diamond\Diamond}a\leq\Diamond a$ and let $x$ be a filter. To prove that $\diamondvert x\subseteq{\diamondvert\diamondvert}x$, let $e\in\diamondvert x$, so that $\Diamond a\leq e$, for some $a\in x$. Then ${\Diamond\Diamond}a\in{\diamondvert\diamondvert}x$ and given ${\Diamond\Diamond}a\leq\Diamond a\leq e$ we obtain $e\in{\diamondvert\diamondvert}x$.

Assume $\Box a\leq a$ and let $y$ be an ideal. Recall that $\boxminus y=\bigvee_{a\in y}y_{\Box a}$, where we use $y_b$ to stand for the principal ideal generated by the lattice element $b$. To prove $\boxminus y\subseteq y$, let $e\in\boxminus y$, so that $e\leq\Box a$ for some $a\in y$. By the hypothesis $\Box a\leq a$ and since $y$ is an ideal we obtain $\Box a\in y$ and then also $e\in y$.

Assume $\Box a\leq{\Box\Box}a$ and let $y$ be an ideal. To prove $\boxminus y\subseteq{\boxminus\boxminus}y$ let $e\in\boxminus y$. Then $e\leq\Box a$ for some $a\in y$. Using the definition of the point operator $\boxminus$, we have $\Box a\in\boxminus y$ and ${\Box\Box}a\in{\boxminus\boxminus}y$. But then we have $e\leq\Box a\leq{\Box\Box}a\in{\boxminus\boxminus}y$, which is an ideal, so $e\in{\boxminus\boxminus}y$.
\end{proof}

\begin{prop}[Canonicity of the B-Axioms]
The logic $\mathbf{\Lambda}=\mathbf{\Lambda}^\ra_{\Box\Diamond}+\{\varphi\proves{\Box\Diamond}\varphi, {\Diamond\Box}\varphi\proves\varphi\}$ is canonical.
\end{prop}
\begin{proof}
Assume that $\Diamond a\leq b$ iff $a\leq \Box b$ holds in an implicative modal algebra and let $\mathfrak{F}^+$ be the full complex algebra of its dual frame. We show that for any filter $x$ and ideal $y$ the condition $x\upv\boxminus y$ iff $\diamondvert x\upv y$, i.e. axiom (FB), holds in the canonical frame.

If $e\in x\cap\boxminus y$, let $a\in y$ such that $e\leq\Box a\in\boxminus y$. By residuation in the modal lattice $\Diamond e\leq a$, but since $e\in x$, then $\diamondvert e\in\diamondvert x$, which is a filter, hence also $a\in x$. But then $a\in\diamondvert x\cap y\neq\emptyset$, i.e. $\diamondvert x\upv y$.

Conversely, if $a\in\diamondvert x\cap y$, let $e\in x$ be such that $\diamondvert e\leq a$. By residuation this is equivalent to $e\leq\Box a$ and since $a\in y$ we get $\Box a\in\boxminus y$. Since $e\in x$, we have $\Box a\in x\cap\boxminus y\neq\emptyset$, i.e. $x\upv\boxminus y$.
\end{proof}

\begin{prop}[Canonicity of Axiom S5]
For any filter $x$, ${\boxminus\diamondvert}x\subseteq\diamondvert x$ and $\boxminus x\subseteq\diamondvert\boxminus x$. In other words, the frame axiom (FS5) is valid in the canonical frame.
\end{prop}
\begin{proof}
If $e\in{\boxminus\diamondvert}x$, then there is some element $c\in\diamondvert x$ such that $\Box c\leq e$, by definition of $\boxminus u=\bigvee_{c\in u}x_{\Box c}$. Also, for $c$ to be in $\diamondvert x=\bigvee_{a\in x}x_{\Diamond a}$ there must be an $a\in x$ such that $\Diamond a\leq c$. For such $a$ we have ${\Box\Diamond} a\leq\Box c\leq e$. But the lattice satisfies $\Diamond a\leq{\Box\Diamond}a$ hence $\Diamond a\leq e$. Since $a\in x$, $\Diamond a\in\diamondvert x$, thereby also $e\in\diamondvert x$.

The argument for $\boxminus x\subseteq\diamondvert\boxminus x$ is similar, left to the reader.
\end{proof}

Note that for a stable set $A$, if $x\in A$, then ${\boxminus\diamondvert} x\in\lbb\ldvert A$, by Lemma~\ref{5lemma}. Then the argument in Proposition~\ref{5Prop} establishes that for any stable set $A$ of filters the inclusions $\ldd A\subseteq \lbb\ldd A$ and $\ldd\lbb A\subseteq\lbb A$ hold and this shows that the axioms S5 are canonical.

\section{Distributive and Heyting Frames}
\label{distributive/Heyting section}
By a {\em distributive frame} we mean a frame  $\mathfrak{F}=(s,Z,I,(R_j)_{j\in J},\sigma)$  such that its complete lattice of stable sets $\gpsi$ is (completely) distributive. Proposition~\ref{upper bound rel prop} identifies a condition for the frame to be distributive. That the canonical frame of a distributive lattice satisfies the condition in this Proposition was proven in \cite[Proposition~5.4, Case 5]{choiceFreeStLog}.

\begin{prop}[\mbox{\cite[Proposition~3.13]{choiceFreeHA}}]
\label{upper bound rel prop}
Let $\mathfrak{F}=(s,Z,I,(R_j)_{j\in J},\sigma)$ be a frame and $\gpsi$ the complete lattice of stable sets.   If all sections of the Galois dual relation $R'_\leq$ of the upper bound relation $R_\leq$ (where $uR_\leq xz$ iff both $x\preceq u$ and $z\preceq u$) are Galois sets, then $\gpsi$ is completely distributive.\telos
\end{prop}

The condition for a frame to be a Heyting frame is that its full complex algebra $\mathfrak{F}^+=(\gpsi,\subseteq,\bigcap,\bigvee,\emptyset'',Z_1,\Ra,\ldots)$ be a (complete) Heyting algebra (with additional structure due to the modal operators). The implication construct $\Ra$ is residuated with a product operator $\bigovert$ (Proposition~\ref{Ra long and short}), which is the closure of the image operator $\bigodot$ generated by the relation $R^{111}$ of Definition~\ref{derived relations defn}. Therefore, $\mathfrak{F}^+$ is a residuated lattice. Hence it will be a Heyting algebra if $\bigovert$ is the same as intersection on stable sets. Working out the detail of this observation we have the following result.

\begin{thm}\label{Heyting thm}
  Let $\mathfrak{F}=(s,Z,I,T^{\partial 1\partial},\ldots)$ be a frame. Then the full complex algebra of the frame is a complete Heyting algebra iff the relation $R^{111}$  coincides with the upper bound relation $R_\leq$ on the frame. In that case, the implication construct is equivalently defined by $x\in (A\Ra C)$ iff $\forall z\in Z_1(x\leq z\lra(z\in A\lra z\in C))$. 
\end{thm}
\begin{proof}
The claims were proven in \cite[Propositions~3.15, 3.17]{choiceFreeHA}. 
\end{proof}   
That the canonical frame of a Heyting algebra satisfies the condition of the above theorem was proven in~\cite[Proposition~4.11]{choiceFreeHA}.

\begin{rem}
When $R^{111}$ coincides with the upper bound relation $R_\leq$, then $\bigodot$ is intersection of stable and, more generally, of upper closed sets in the $\leq$ order. This follows from the fact that $U\bigodot W=\{u\midsp\exists x,z(x\in U\wedge z\in W\wedge uR_\leq xz)\}$. But note that $\bigodot$ need not be intersection of arbitrary subsets of $Z_1$. Note also that for any increasing subsets (any upsets) $U,W$, the set $U\Ra_T W$ is also an upset, when $R^{111}=R_\leq$. Hence, in the case of a Heyting frame, the residuated pair $\bigodot,\Ra_T$ in $\powerset(Z_1)$ restricts to a residuated pair in the family (complete lattice) $\mathcal{UP}(Z_1)$ of increasing subsets of $Z_1$ (upper closed subsets in the $\preceq$-order), and to a residuated pair $\bigodot,\Ra$ in $\gpsi$. The latter is the case because for stable sets $A,C$, $A\bigodot C=A\cap C$ is a stable set hence $A\bigovert C=(A\bigodot C)''=A\bigodot C=A\cap C$.
\end{rem}

\begin{thm}\label{heyting K}
If the frame $\mathfrak{F}=(s,Z,I,T^{\partial 1\partial},R^{\partial\partial}_\Box,\ldots)$ is a Heyting frame, then the K-axiom is valid in $\gpsi$. Consequently, the minimal Intuitionistic modal logic, assuming only the K-axiom, is valid in the class of Heyting frames.
\end{thm}
\begin{proof}
Assume $u\in\lbb(A\Ra C)$. To show that $u\in\lbb A\Ra\lbb C$, let $z\in Z_1$ be such that $u\leq z$ and $z\in\lbb A$. To get $z\in\lbb C$, which is equivalent to $zR''_\Box\subseteq C$, let $p\in Z_1$ be such that $zR''_\Box p$. It suffices to get $p\in C$. From $u\leq z, zR''_\Box p$ and Lemma~\ref{monotonicity props of R double prime} it follows that $uR''_\Box p$. Using the hypothesis $u\in\lbb(A\Ra C)$ we then have $p\in A\Ra C$. Furthermore, using the hypothesis $z\in\lbb A$ we also get $p\in A$. Given the definition of $\Ra$ in the Heyting case, $A\Ra C=\{p\midsp\forall z(p\leq z\wedge z\in A\lra z\in C)\}$, and since $p$ suttisfies the premisses of the defining implication for $A\Ra C$, we get $p\in C$, as needed. Alternatively, use the fact that $p\in A\cap(A\Ra C)\subseteq C$.
\end{proof}

By Theorem~\ref{heyting K}, no need arises in the case of a modal extension of the Intuitionistic propositional calculus to consider second-order axioms for frames. Frames for extensions with any of the axioms D, T, B, S4, or S5 can be dealt with roughly as done for the distribution-free case, but we shall not dwell on this here.

\section{Further Issues and Conclusions}
This article is part of an ongoing project. It worked through the necessary background, presenting methods and techniques that can be used to model distribution-free systems, applied to the case of modal logic. The area is largely unexplored and much remains to be done, the main interest being in finding out how much of the classical model theory of normal modal logics on a Boolean propositional basis can be actually carried over, appropriately modified, to the case of distribution-free modal logic. We outline some of the related issues.

With a notion of {\em weak bounded morphism}, introduced in \cite{duality2}, the duality of \cite{duality2} can be specialized to a duality for implicative modal lattices, as it has been already specialized to dualities for implicative lattices and Heyting algebras \cite{choiceFreeHA} and for logics with a weak negation logical operator \cite{choiceFreeStLog}. 

Introducing both a classical but sorted modal logic of frames and a sorted first-order language, the van Benthem correspondence problem for distribution-free modal logics can be resolved, as it is an instance of the problem for the logics of arbitrary normal lattice expansions, a project that has been carried out in \cite{vb}. 

The problem of calculating first-order correspondents can be reduced to the same problem but for classical, though sorted, modal logic. 
The sorted modal language of a frame is the language of the sorted powerset algebra $\mathbf{P}=(\largediamond:\powerset(Z_1)\leftrightarrows\powerset(Z_\partial):\lbbox, (F_j)_{j\in J} )$, displayed below for the case of interest in this article.
\begin{eqnarray*}
\mathcal{L}_1\ni\alpha,\eta,\zeta &=& P_i(i\in\mathbb{N})\midsp p_i(i\in\mathbb{N})\midsp\neg\alpha\midsp\alpha\cup\alpha\midsp{\diamondvert}\alpha\midsp{\bbox}\beta  \\
\mathcal{L}_\partial\ni\beta,\delta,\xi &=& Q_i(i\in\mathbb{N})\midsp q_i(i\in\mathbb{N})\midsp\neg\beta\midsp\beta\cup\beta\midsp{\diamondminus}\beta\midsp \alpha{\tright}\beta\midsp{\Box}\alpha
\end{eqnarray*}
The language is interpreted in frames as displayed in Table~\ref{sorted sat table}.

\begin{table}[!htbp]
\caption{Sorted satisfaction relation}
\label{sorted sat table}
\begin{center} ($u\in Z_1, v\in Z_\partial$ and
 $v\onto{x}y$ is alternative notation for $vT^{\partial 1\partial}xy$)\end{center}
\begin{tabbing}
$u\models P_i$\hskip7mm\=iff\hskip2mm\= $u\in V_1(P_i)$\hskip2.2cm\= $v\vmodels Q_i$ \hskip6mm\=iff\hskip2mm\= $v\in V_\partial(Q_i)$
\\[2mm]
$u\models p_i$ \> iff \> $u\in (V_1(P_i))''$ \> $v\vmodels q_i$ \>iff\> $v\in (V_1(P_i))'$
\\[2mm]
$u\models\neg\alpha$\>iff\> $u\not\models\alpha$    \>  $v\vmodels\neg\beta$ \>iff\> $v\not\vmodels\beta$
\\[2mm]
$u\models\alpha\cup\eta$ \>iff\> $u\models\alpha$ or $u\models\eta$
    \>
    $v\vmodels\beta\cup\delta$ \>iff\> $v\vmodels\beta$ or $v\vmodels\delta$
\\[2mm]
$u\models{\bbox\beta}$ \>iff\> $\forall y\in Z_\partial(uIy\lra y\vmodels\beta)$
    \>
    $v\vmodels{\Box}\alpha$ \>iff\> $\forall x\in Z_1(xIy\lra x\models\alpha)$
\\[2mm]
$u\models\diamondvert\alpha$ \>iff\>    $\exists z\in Z_1(uR^{11}_\Diamond z\wedge z\models\alpha)$
    \>
    $v\vmodels\diamondminus\beta$ \>iff\>   $\exists y\in Z_\partial(vR^{\partial\partial}_\Box y\wedge y\vmodels\beta)$
\\[4mm]
\hskip4cm\=$v\vmodels \alpha{\tright}\beta$ \hskip7mm\= iff\hskip2mm\=  $\exists x\in Z_1\exists y\in Z_\partial(x\models\alpha\wedge y\vmodels\beta\wedge (v\onto{x}y))$
\end{tabbing}
\hrulefill
\end{table}

Including a set of {\em regular variables} $p_i, q_i$ in the language, interpreted as Galois sets,  allows for regarding the language of distribution-free modal logic literally as the sublanguage of {\em regular}, or {\em stable} sentences, where $\alpha''\ra\alpha$ is valid,
\[
\mathcal{L}_r\ni\varphi,\psi,\vartheta,\chi\;=\;p_i(i\in\mathbb{N})\midsp\ufootl\midsp\dfootl\midsp \varphi\wedge\varphi\midsp\varphi\vee\varphi\midsp\varphi\circ\varphi\midsp\varphi\rfspoon\varphi\midsp{\diamonddiamond}\varphi \midsp{\boxbox}\varphi,
\]
assuming the definitional axioms $\bb\alpha\leftrightarrow(\diamondminus \alpha')'$, $\dd\alpha\leftrightarrow(\diamondvert\alpha)''$, $\alpha\rfspoon\eta\leftrightarrow(\alpha{\tright}\eta')'$, where $\alpha'\leftrightarrow{\Box}\neg\alpha$, $\beta'\leftrightarrow{\bbox}\neg\beta$ and where also it is defined that $\ufootl\leftrightarrow\top$, $\dfootl\leftrightarrow\bot''$. 

For example, despite its definition, it has been shown that the satisfaction clause for $\bb$ is $x\forces\bb\varphi$ iff $\forall z(xR''_\Box z\lra z\forces\varphi)$. Having shown that in frames the dual operator $\lbb$ is simply the restriction of a classical dual operator $\lbminus_R$, generated by the relation $R''_\Box$, we might as well introduce $\bb$ in the syntax of the sorted modal language as primitive, then include an axiom to the effect that $(\alpha''\ra\alpha)\ra((\bb\alpha)''\ra\bb\alpha)$. 

We sketched above, admittedly too roughly, how the correspondence problem for distribution-free modal logic can be reduced to the classical problem, the only difference in this setting being that the language is sorted. We leave further details for another report.

A Goldblatt-Thomason theorem has been proven by Goldblatt \cite{goldblatt-morphisms2019} for non-distributive logics, though only operators and dual operators have been explicitly treated. It should not be hard to lift the classical theorem to the sorted modal logic case and we conjecture that this can lead to deriving a related result for distribution-free modal logics as well, though clarifying the detail of this will have to be carried out in a subsequent report.

\bibliographystyle{plain}

\end{document}